\documentclass[graybox]{svmult}

\usepackage{type1cm}        
\usepackage{makeidx}         
\usepackage{graphicx}        
\usepackage{multicol}        
\usepackage[bottom]{footmisc}

\usepackage{newtxtext}       %
\usepackage[varvw]{newtxmath}       

\usepackage{bbm}

\usepackage{mathrsfs,bm,enumerate}
\usepackage{caption,subcaption}
\usepackage[T1]{fontenc} 
\usepackage[latin9]{inputenc}
\usepackage{amsthm}
\usepackage{mathtools}
\usepackage{bookmark}

\newcommand{\inR}{\in \mathbb{R}}
\newcommand{\inZ}{\in \mathbb{Z}}
\newcommand{\C}{ \mathbb{C}}
\newcommand{\R}{ \mathbb{R}}
\newcommand{\Z}{ \mathbb{Z}}

\newcommand{\N}{ \mathbb{N}}

\newcommand{\eqdef}{\stackrel{\vartriangle}{=}}

\newcommand{\dint}{{\rm d}}

\newcommand{\Fourier}{ \mathcal{F}}

\newcommand{\bw}{{\boldsymbol \omega}}

\newcommand{\bk}{{\boldsymbol k}}

\DeclareMathOperator*{\esssup}{ess\,sup}

\DeclareMathOperator*{\prox}{prox}
\DeclareMathOperator*{\argmin}{argmin}

\newcommand{\Tr}{\mathsf{T}}
\newcommand{\HTop}{\mathsf{H}}

\def\V#1{{\boldsymbol{#1}}}         
\def\Spc#1{{\mathcal{#1}}}  
\def\M#1{{\bf{#1}}}  
\def\Op#1{{\mathrm{#1}}}  
\def\ee{\mathrm{e}} 
\def\jj{\mathrm{j}} 
\def\Indic{\mathbbm{1}}

\def\Identity{\mathrm{Id}} %

\newcommand{\toC}{\xrightarrow{\mbox{\tiny \ \rm c.  }}}
\makeindex             

\begin{document}

\title*{Parseval Convolution Operators\\ and Neural Networks}
\author{Michael Unser\orcidID{0000-0003-1248-2513} and\\ Stanislas Ducotterd\orcidID{0009-0006-2047-5179}}
\institute{Michael Unser \at EPFL, Lausanne, Switzerland, \email{michael.unser@epfl.ch}
\and Stanislas Ducotterd \at EPFL, Lausanne, Switzerland, \email{stanislas.Ducotterd@epfl.ch}}
%
%
\maketitle

\abstract*{We first establish a kernel theorem that characterizes all linear shift-invariant (LSI) operators acting on discrete multicomponent signals. This result naturally leads to the identification of the Parseval convolution operators as the class of energy-preserving filterbanks. We then present a constructive approach for the design/specification of such filterbanks via the chaining of elementary Parseval modules, each of which being parameterized by an orthogonal matrix or a 1-tight frame. Our analysis is complemented with explicit formulas for the Lipschitz constant of all the components of a convolutional neural network (CNN), which gives us a handle on their stability. Finally, we demonstrate the usage of those tools with the design of a CNN-based algorithm for the iterative reconstruction of biomedical images. Our algorithm falls within the plug-and-play framework for the resolution of 
inverse problems. It yields better-quality results than the 
sparsity-based methods used in compressed sensing, while offering essentially the same convergence and robustness guarantees.}

\abstract{We first establish a kernel theorem that characterizes all linear shift-invariant (LSI) operators acting on discrete multicomponent signals. This result naturally leads to the identification of the Parseval convolution operators as the class of energy-preserving filterbanks. We then present a constructive approach for the design/specification of such filterbanks via the chaining of elementary Parseval modules, each of which being parameterized by an orthogonal matrix or a 1-tight frame. Our analysis is complemented with explicit formulas for the Lipschitz constant of all the components of a convolutional neural network (CNN), which gives us a handle on their stability. Finally, we demonstrate the usage of those tools with the design of a CNN-based algorithm for the iterative reconstruction of biomedical images. Our algorithm falls within the plug-and-play framework for the resolution of 
inverse problems. It yields better-quality results than the 
sparsity-based methods used in compressed sensing, while offering essentially the same convergence and robustness guarantees.}

\section{Introduction}
\label{Sec:Intro}
The goal of this chapter is twofold.  The first objective is to characterize a special type of convolutional operators that are 
robust and inherently stable because of their Parseval property. The second objective is to showcase the use of these operators in the design of thrustworthy neural-network-based algorithms for signal and image processing.
Our approach is deductive, in that it relies on the higher-level tools of functional analysis 
to identify the relevant operators based on their fundamental properties; namely, linearity, shift-invariance (LSI), and energy conservation (Parseval).

The study of
LSI operators (a.k.a.\ filters) relies heavily on the Fourier transform and is a central topic in linear-systems theory and signal processing  \cite{Kailath1980,Oppenheim1999,Vetterli2014}. Hence, the first step of our investigation is to extend the classic framework to accommodate the kind of processing performed in convolutional neural networks (CNNs), where the convolutional layers have multichanel inputs and outputs.
We do so by adopting an operator-based formalism with appropriate Hilbert spaces, which then also makes the description of CNNs mathematically precise.
As one may expect, the corresponding LSI operators are characterized  by their impulse response or, equivalently, by their frequency response, the extension to the classic setting of signal processing being that these entities now both happen to be matrix-valued (see Theorem  \ref{Theo:L2VecBound}). 

Our focus on Parseval operators is motivated by the desire to control the stability of the components of CNNs, which can be quantified mathematically by their Lipschitz constant (see Section \ref{Sec:ParsevalOperator}). Indeed, it is known that the stability of conventional deep neural networks degrades (almost) exponentially with their depth \cite{Zou2019}. This lack of stability partly explains why CNNs can occasionally hallucinate, which is unacceptable for critical applications such as, for instance, diagnostic imaging. Our proposed remedy is to constrain the Lipschitz constant of each layer, with the ``ultra-stable'' configuration being the one where each component is non-expansive; i.e., with a Lipschitz constant no greater than $1$.
Parseval operators are exemplar in this respect since they preserve energy, which, in effect, turns the (worst-case) Lipschitz bound into an equality. Additional features that motivate their usage are as follows:
\begin{enumerate}
\item Parseval convolution operators have a remarkably simple theoretical description, which is given in Proposition \ref{Theo:PasevalFilters};
\item they admit convenient parametric representations (see Section \ref{Sec:Parameterization}) that are directly amenable to an optimization in standard computational frameworks for machine learning, such as PyTorch.
\end{enumerate}
However, one must acknowledge that there is no free lunch. Any attempt to stabilize a neural network by
constraining the Lipschitz constant of each layer will necessarily reduce its expressivity, as documented in \cite{Huster2019,Ducotterd2024}. The good news is that this effect is less pronounced when the linear layers have the Parseval property, as confirmed in our experiments (Parseval vs. spectral normalization). In fact, we shall demonstrate that the use of Parseval CNNs (as substitute for the classic proximity operator of convex optimization) results in a substantial improvement in the quality of image reconstruction over that of the traditional sparsity-based methods of compressed sensing, while it offers essentially  the same theoretical guarantees (consistency and stability).

\subsection{Related Works and Concepts}
Conceptually, Parseval operators are the infinite-dimensional generalization of orthogonal matrices (one-to-one scenario) and, more 
generally, of $1$-tight frames (one-to-many scenario) \cite{Christensen1995}. The latter involve rectangular matrices $\M A \in \R^{M \times N}$ with $M\ge N$ and the property that $\M A^\Tr \M A=\M I_N$ (identity).
When the operator is LSI, then it is diagonalized by the Fourier transform---a property that can be exploited for the design of Parseval filterbanks.

The specification of filters for the orthogonal wavelet transform \cite{Daubechies1992,Mallat1998,Meyer1990} is a special instance of the one-to-one scenario. In fact, there is a comprehensive theory for the design
of perfect-reconstruction filterbanks  \cite{Strang:1996,Vetterli:1995}, the orthogonal ones being sometimes referred to as lossless systems \cite{Vaidyanathan.93}. It includes general factorization results for paraunitary matrices associated with finite impulse response (FIR) filterbanks of a given McMillan degree \cite[Theorem 14.4.1, p. 736]{Vaidyanathan.93} or of a given size \cite{Turcajova1994b}, with the caveat that these only hold in the one-dimensional setting. There are also adaptations of those results for linear-phase filters \cite{Soman1993,Tran2000,Turcajova1994}.

The one-to-many scenario (wavelet frames) caught the interest of researchers in the late '90s, motivated by an early application in texture analysis \cite{Unser1995d} that involved a computational architecture that is a ``handcrafted'' form
of CNN. Such redundant wavelet designs are less constrained than the orthogonal ones. They go under the name of
oversampled filterbanks \cite{Cvetkovic1998}, oversampled wavelet transforms \cite{Bolcskei1998}, undecimated wavelet transform \cite{Luisier2011b}, lapped transforms \cite{Chebira2007}, or, more generally,
tight (wavelet) frames \cite{Aldroubi1995,Christensen2003,Kovacevic2007b,Kovacevic2007,Kovacevic2007a}.

The use of Parseval operators in the context of neural networks is more recent. The new twist brought forth by machine learning is that the filters can now be learned to provide the best performance for a given computational task, which is feasible under the availability of sufficient training data and computational power. The first attempts to orthogonalize the linear layers of a neural network were motivated by the desire to avoid vanishing gradients and to improve robustness against adversarial attacks \cite{Anil_PMLR2019,Cisse2017,Hasannasab2020,Xiao_dynamical_2018}. Several research teams  \cite{HuangCVPR2020,Li2019,Su2022} then proposed solutions for the training of orthogonal convolution layers that are inspired by the one-dimensional factorization theorems uncovered during the development of the wavelet transform. There are also approaches
that operate directly in the Fourier domain \cite{Trockman2021}.

\subsection{Road Map}
This chapter is organized as follows.

We start with a presentation of background material in Section \ref{Sec:Background}. First, we set the notation and introduce the Hilbert spaces for the representation of $d$-dimensional vector-valued signals, such as $\ell_2^N(\Z^d)$, whose elements are $N$-component  discrete signals or images. We then move on to the discrete Fourier transform in Section \ref{Sec:DiscretFourier}. This is complemented with a discussion of fundamental continuity/stability properties of operators in the general context of Banach/Hilbert spaces in Section \ref{Sec:ParsevalOperator}.

In Section \ref{Sec:VecValued}, we focus on the discrete LSI setting and identify the complete family of continuous LSI operators $\Op T_{\rm LSI}:\ell_2^N(\Z^d) \to \ell_2^M(\Z^d)$ (Theorem \ref{Theo:L2VecBound}), including the determination of their Lipschitz constant. These operators
are multichannel filters with $N$-channel inputs and $M$-channel outputs. They are uniquely specified by their {\em matrix-valued impulse response} or, equivalently, by their {\em matrix-valued frequency response}. An important subclass are the LSI-Parseval operators; these are identified in Proposition \ref{Theo:PasevalFilters} as convolution operators with a {\em paraunitary} frequency response. 

In Section \ref{Sec:Parameterization}, we develop a constructive approach for the design/specification of Parseval filterbanks. The leading idea is to generate higher-complexity filterbanks through the chaining of elementary Parseval modules, each 
being parameterized by a unitary matrix or, eventually, a $1$-tight frame (see Table \ref{Tab:ParsevalFilters}).

Section \ref{sec:PnPReconstruction} is devoted to the application of our framework to the problem of biomedical image reconstruction. Our approach revolves around the design of
a robust $1$-Lipschitz CNN for image denoising that mimics the architecture of DnCNN \cite{ZZCMZ2017}---a very popular  image denoiser. The important twist is that, unlike DnCNN, the convolution layers of our network are constrained to be Parseval, which makes our denoiser compatible with the powerful plug-and-play (PnP) paradigm for the resolution of linear inverse problems \cite{Chan2016plug,Kamilov2023plug,Sun2021,Venkatakrishnan2013plug}. We first provide mathematical support for this procedure in the form of convergence guarantees and stability bounds. We then demonstrate the feasibility of the approach for MRI reconstruction and report experimental results where it significantly outperforms the standard technique (total-variation-regularized reconstruction) used in compressed sensing.

\section{Mathematical Background}
\label{Sec:Background}
\subsection{Notation}
We use boldface lower and upper case letters to denote vectors and matrices, respectively (e.g., $\M u \in \R^N$ and $\M U \in \C^{M \times N}$). Specific instances are  $\M e_n$ (the $n$th element of the canonical basis in $\R^N$) and $\M I_N=[\M e_1 \dots \M e_N]$ (the unit matrix of size $N$).

A discrete multidimensional scalar signal (e.g., the input or output of a convolutional neural network) is a sequence $(x[\bk])_{\bk \in \Z^d}$ of real numbers that, depending on the context, will be denoted as $x \in \ell_2(\Z^d)$ (i.e., as a member of a Hilbert space),
or $x[\cdot]$, where ``$\cdot$''  is a placeholder for the indexing variable. Our use of square brackets follows the convention of signal processing, as 
reminder of the discrete nature of the objects. A vector-valued signal
$\M x[\cdot]=(x_1[\cdot],\dots,x_N[\cdot])$ 
is an indexed sequence of vectors $\M x[\bk]=(x_1[\bk],\dots,x_N[\V k])\in \R^N$ with $\bk$ ranging over $\Z^d$.
Likewise, $\M X[\cdot]=\begin{bmatrix}\M x_1[\cdot] &\cdots &\M x_M[\cdot]\end{bmatrix}$ is a matrix-valued signal or sequence.
An alternative representation of such sequences is 
\begin{align}
\M X[\cdot]=\sum_{\V n \in \Z^d}\M X[\V n]\delta[\cdot- \V n]
\end{align}
 where $\delta[\cdot-\V n]$ denotes the (scalar) Kronecker impulse shifted by $\V n$, with
$\delta[\M 0]=1$ and $\delta[\bk]=0$ for $\bk \in \Z^d \backslash\{\M 0\}$.

In the same spirit, we use the notation $f(\cdot)$, $\M f(\cdot)$, $\M F(\cdot)$
to designate objects that are respectively scalar, vector-valued, and matrix-valued functions of a continuously varying index such as $\V x \in \R^d$ or $\bw \in \mathbb{T}^d=[-\pi,+\pi]^d$ (the frequency variable).

The symbol $^\vee$ denotes the flipping operator with $\M x^\vee[\bk]\eqdef\M x[-\bk]$ for all $\bk \in \Z^d$, while $\M U^\HTop \in \C^{N \times M}$
is the Hermitian transpose of the complex matrix $\M U \in \C^{M \times N}$ with $[\M U^\HTop]_{n,m}\eqdef[\overline{\M U}]_{m,n}$ (transpose with complex conjugation). 

Our primary Hilbert space of vector-valued signals is $\ell^N_2(\Z^d)=\ell_2(\{1,\dots,N\} \times \Z^d )=\ell_2(\Z^d) \times \dots \times \ell_2(\Z^d)$ \big($N$ occurrences of $\ell_2(\Z^d)$\big), which is the direct-product extension of $\ell_2(\Z^d)$. Specifically,
$$\ell^N_2(\Z^d)=\left\{\M x[\cdot]: \Z^d \to \R^N \mbox{ s.t. }  \|\M x [\cdot]\|_{\ell^N_2(\Z^d)} <\infty\right\}$$
with 
\begin{align}
\label{Eq:lenorm}
\|\M x [\cdot]\|_{\ell^N_2(\Z^d)}\eqdef 
\left(\sum_{n=1}^N \sum_{\bk \inZ^d} |x_n[\bk]|^2\right)^{1/2}.
\end{align}
By invoking a density argument,
we can interchange the order of summation in \eqref{Eq:lenorm} so that
$$\|\M x [\cdot]\|_{\ell^N_2(\Z^d)}=
\big\| ( \|x_1\|_{\ell_2(\Z^d)},\dots,\|x_N\|_{\ell_2(\Z^d)})\big\|_{2}=\big\| \|\M x [\cdot]\|_2\big\|_{\ell_2(\Z^d)},$$
where 
$$
\|\M u\|_2\eqdef \displaystyle\left(\sum_{n=1}^N |u_n|^2\right)^{1/2}$$
is the conventional Euclidean norm of the vector $\M u=(u_1,\dots,u_N)\in \C^N$.

\subsection{The Discrete Fourier Transform and Plancherel's Isomorphism}
\label{Sec:DiscretFourier}
The discrete Fourier transform of a signal $x[\cdot] \in \ell_1(\Z^d) \subset \ell_2(\Z^d)$ is  defined as
\begin{align}
\label{Eq:Fourierl1}
\hat x(\bw)\eqdef \Fourier_{\rm d}\{x[\cdot]\}(\bw)=\sum_{\bk \inZ^d} x[\bk] \ee^{-\jj\langle \bw, \bk\rangle},\quad \bw \in \R^d.
\end{align}
The function $\hat x: \R^d \to \C$ is continuous, bounded
and $2\pi$-periodic. It is therefore entirely specified by its main period $\mathbb{T}^d=[-\pi,\pi]^d$.
The original signal can be recovered by inverse Fourier transformation as
\begin{align}
\label{Eq:InvFourier}
x[\V k]=\Fourier^{-1}_{\rm d}\{\hat a\}[\V k]=\int_{\mathbb{T}^d} \hat x(\bw) \ee^{\jj\langle \bw, \bk\rangle} \frac{\dint \bw}{(2 \pi)^d},\quad \V k\inZ^d.
\end{align}
By interpreting the infinite sum in \eqref{Eq:Fourierl1} as an appropriate limit, one  then extends the definition of the Fourier transform to encompass all square-summable signals. This yields the extended operator
$\Fourier_{\rm d}: \ell_2(\Z^d) \to L_2(\mathbb{T}^d)$, where $L_2(\mathbb{T}^d)$
is the space of measurable complex Hermitian-symmetric and square-integrable functions on $\mathbb{T}^d$. The latter is a Hilbert space equipped with the Hermitian inner product
\begin{align}
\langle \hat x, \hat y \rangle_{L_2({\mathbb{T}^d})}
\eqdef \int_{\mathbb{T}^d} \hat x(\bw) \overline{\hat y(\bw)} \;\frac{\dint \bw}{(2 \pi)^d}.
\end{align}
The 
Fourier transform $\Fourier_{\rm d}: \ell_2(\Z^d) \to L_2(\mathbb{T}^d)$ is a bijective isometry (unitary map between two Hilbert spaces)
with $\Fourier^{-1}_{\rm d}: L_2(\mathbb{T}^d) \to \ell_2(\Z^d)$, where the inverse transform is still specified by \eqref{Eq:InvFourier} with an extended Lebesgue interpretation of the integral. 
Indeed, by invoking the Cauchy-Schwartz inequality, we get that
$$\int_{\mathbb{T}^d} \left|\hat x(\bw) \ee^{\jj\langle \bw, \bk\rangle} \right|\frac{\dint \bw}{(2 \pi)^d}= \int_{\mathbb{T}^d} |\hat x(\bw)| \times 1 \;\frac{\dint \bw}{(2 \pi)^d}\le \|\hat x\|_{L_2} \left(\int_{\mathbb{T}^d} 1 \;\frac{\dint \bw}{(2 \pi)^d}\right)^\frac{1}{2}=\|\hat x\|_{L_2},$$
which ensures the well-posednessed of \eqref{Eq:InvFourier} for all $\hat x \in L_2(\mathbb{T}^d)$.

The cornerstone of the $\ell_2$ theory of the Fourier transform is the Plancherel-Parseval identity
\begin{align}
\label{Eq:Plancherel}
\forall x, y \in \ell_2(\Z^d): \quad \langle x, y \rangle_{\ell_2(\Z^d)}\eqdef\sum_{\bk \in \Z^d} x[\bk]\overline{y[\bk]}=\langle \hat x, \hat y \rangle_{L_2(\mathbb{T}^d)}.
\end{align}
It ensures that the inner product is preserved in the Fourier domain.

The vector-valued extension of these relations is immediate if one defines
the Fourier transform of a vector-valued signal $\M x[\cdot] \in \ell_2^N(\Z^d)$ as
\begin{align}
\label{Eq:FourierVec}
\forall \bw \in \R^d: \Fourier_{\rm d}\{\M x[\cdot]\}(\bw)=\Fourier_{\rm d}\left\{
\left[
\begin{array}{c}
x_1[\cdot]  \\
  \vdots      \\
x_N[\cdot]
\end{array}
\right]\right\}(\bw)\eqdef \left[
\begin{array}{c}
\hat x_1(\bw)   \\
  \vdots      \\
\hat x_N(\bw) 
\end{array}
\right]=\widehat{ \M x}(\bw)
\end{align}
and its inverse as
\begin{align}
\label{Eq:FourierInvVec}
\Fourier^{-1}_{\rm d}\{\widehat{\M x}\}
=\Fourier^{-1}_{\rm d}\left\{
\left[
\begin{array}{c}
\hat x_1   \\
  \vdots      \\
\hat x_N
\end{array}
\right]\right\}
=
\left[
\begin{array}{c}
\Fourier^{-1}_{\rm d}\{\hat x_1\}   \\
  \vdots      \\
\Fourier^{-1}_{\rm d}\{\hat x_N\}\end{array}
\right]=\left[
\begin{array}{c}
x_1[\cdot]   \\
  \vdots      \\
x_N[\cdot] \end{array}
\right].
\end{align}
The corresponding vector-valued version of Plancherel's identity reads
\begin{align}
\forall \M x, \M y \in \ell^N_2(\Z^d): \quad \langle \M x, \M y \rangle_{\ell^N_2(\Z^d)}&\eqdef\sum_{n=1}^N \sum_{\bk \in \Z^d} x_n[\bk]\overline{y_n[\bk]} \nonumber\\
&=\langle \widehat{\M x}, \widehat{\M y} \rangle_{L^N_2(\mathbb{T}^d)}\eqdef\sum_{n=1}^N\int_{\mathbb{T}^d} \hat x_n(\bw) \overline{\hat y_n(\bw)} \;\frac{\dint \bw}{(2 \pi)^d}.
\label{Eq:L2innerprod}
\end{align}
The Plancherel-Fourier isomorphism is then expressed as $\Fourier_{\rm d}: \ell_2^N(\Z^d) \to 
L_2^N(\mathbb{T}^d)$ and $\Fourier^{-1}_{\rm d}: L_2^N(\mathbb{T}^d) \to \ell_2^N(\Z^d)$ where $L_2^N(\mathbb{T}^d)$ is the Hilbert space of complex-valued Hermitian-symmetric functions associated with the inner product \eqref{Eq:L2innerprod} for $(2\pi)$-periodic vector-valued functions.
\subsection{1-Lip and Parseval Operators}
\label{Sec:ParsevalOperator}
The transformations that occur in a neural network can be described
through the action of some operators $\Op T$ that map any member $x$ of a vector space $\Spc X$ (for instance, a specific input of the network or of one of its layers) into some element $y=\Op T\{x\}$ of
another vector space $\Spc Y_i$ (e.g., the output of the network or any of its intermediate layers). These operators $\Op T: \Spc X \to \Spc Y$ can be linear (as in the case of a convolution layer) or, more generally, nonlinear. A minimal requirement is that the $\Op T$ be {\em continuous}, which is a 
mathematical precondition 
tied to the underlying topologies.

\begin{definition}
\label{Def:Continuity}
Consider the (possibly nonlinear) mapping $\Op T: \Spc X \to \Spc Y$, where $\Spc X=(\Spc X,\|\cdot\|_{\Spc X})$ and $\Spc Y=(\Spc Y,\|\cdot\|_{\Spc Y})$ are two complete normed spaces (e.g., Banach or Hilbert spaces). Then,
$\Op T$ can exhibit the following forms of continuity.
\begin{enumerate}
\item Continuity at $x_0 \in \Spc X$: For any $\epsilon>0$, there exists some $\mu>0$ such that, for any $x \in \Spc X$ with
$\|x - x_0\|_{\Spc X}< \mu$, it holds that $\|\Op T\{x\}-\Op T\{x_0\}\|_{\Spc Y}< \epsilon$.
\item Uniform continuity on $\Spc X$: For any $\epsilon>0$, there exists some $\mu>0$ such that, for any $x, x_0 \in \Spc X$ with
$\|x - x_0\|_{\Spc X}< \mu$, it holds that $\|\Op T\{x\}-\Op T\{x_0\}\|_{\Spc Y}< \epsilon$.
\item Lipschitz continuity: There exists some constant $L>0$ such that
\begin{align}
\label{Eq:Lipcshitz}
\forall x,x_0 \in \Spc X: \|\Op T\{x\}-\Op T \{x_0\}\|_{\Spc Y}\le L \|x-x_0\|_{\Spc X}.
\end{align}
\end{enumerate}
\end{definition}
\vspace*{1ex}
The third form (Lipschitz) is obviously also the strongest with $3 \Rightarrow 2 \Rightarrow 1$.
The smallest $L$ for which \eqref{Eq:Lipcshitz} holds is called the Lipschitz constant of $\Op T$ with
\begin{align}
\label{Eq:LipschitzNorm}
{\rm Lip}(\Op T)=\sup_{ \forall x,x_0 \in \Spc X,\  x\ne x_0}  \frac{\|\Op T\{x\}-\Op T\{x_0\}\|_{\Spc Y}}{\|x-x_0\|_{\Spc X}}.
\end{align}

\begin{definition}
An operator $\Op T: \Spc X \to \Spc Y$ is said to be of $1$-Lip type if ${\rm Lip}(\Op T)=1$.
\end{definition}
The $1$-Lip operators  are of special interest to us because they are inherently stable: a small perturbation of their input can only induce a small deviation of their output.
Moreover, they can  be chained at will without any degradation in overall stability because ${\rm Lip}(\Op T_2 \circ \Op T_1)\le{\rm Lip}(\Op T_2) {\rm Lip}(\Op T_1)=1$.

For linear operators, the graded forms of continuity in Definition \ref{Def:Continuity} can all be related to one overarching
simplifying concept: the {\em boundedness} of the operator. The two key ideas there are: (i) a linear operator is (locally) continuous at any $x_0\in \Spc X$ if and only if it is continuous at $0$; and, (ii) it is uniformly continuous if and only if it is bounded \cite[Theorem 2.9-2, p. 84]{Ciarlet2013}. 
Finally, there is one very attractive form of $1$-Lip linear operators for which \eqref{Eq:Lipcshitz} holds as an equality, rather than a ``worst-case'' inequality.
%
%
%
%
To make this explicit, we now recall some basic properties of linear operators acting on Hilbert spaces and identify the subclass of Parseval operators, which are norm- as well as inner-product (angle) preserving.
\begin{definition}
\label{Def:Parseval}
Let $\Spc X$ and $\Spc Y$ be two Hilbert spaces.
The most basic Hilbertian properties of a linear operator
$\Op T: \Spc X \to \Spc Y$ are as follows.
\begin{enumerate}
\item Boundedness (continuity): There exists a constant $B<\infty$ such that
\begin{align}
\label{Eq:Boundedup}
\forall x\in \Spc X:\quad \|\Op T\{x\}\|_{\Spc Y} \le B\, \| x\|_{\Spc X},
\end{align}
with the smallest $B$ in \eqref{Eq:Boundedup} being the norm of the operator denoted by
$\|\Op T\|$.
\item Boundedness from below (injectivity): There exists a constant $0<A$ such that
\begin{align}
\label{Eq:Boundedlow}
\forall x\in \Spc X:\quad A\; \| x\|_{\Spc X}\le \|\Op T\{x\}\|_{\Spc Y}.
\end{align}
\item Isometry: $\Op T$ is norm-preserving (or Parseval), meaning that both \eqref{Eq:Boundedup} and \eqref{Eq:Boundedlow} hold with $A=B=1$.
\end{enumerate}
\end{definition}

To identify the critical bounds in Definition \ref{Def:Parseval}, we observe that, for any $x \in \Spc X\backslash\{0\}$,
\begin{align}
\label{Eq:normalisation}
A\le \frac{\|\Op T\{x\}\|_{\Spc Y}}{\|x\|_{\Spc X}}=\|\Op T\{z\}\|_{\Spc Y} \le B \, \mbox{ with } z=\frac{x}{\|x\|_{\Spc X}}.
\end{align}
This holds by virtue of the linearity of $\Op T$ and the homogeneity property of the norm. In particular, this allows us to specify the induced norm of the operator as
\begin{align}
\label{EQ:NormOp}
\|\Op T\|\eqdef \sup_{x \in \Spc X \backslash\{0\}}\frac{\|\Op T\{x\}\|_{\Spc Y}}{\|x\|_{\Spc X}}=\sup_{z \in \Spc X:\  \|z\|_\Spc X=1}\|\Op T\{z\}\|_{\Spc Y}.\end{align}
Note that \eqref{EQ:NormOp} can be obtained by restricting \eqref{Eq:LipschitzNorm} to $x_0=0$, which then also yields $\|\Op T\|={\rm Lip}(\Op T)$ due to the linearity of $\Op T$.
The isometry property (Item 3) is by far the most constraining, as it implies the two others 
with $\|\Op T\|=1$. As it turns out, it has other remarkable consequences, which yield some alternative characterization(s).

\begin{proposition}[Properties of Parseval operators]
\label{Prop:Parseval}
Let $\Spc X$ and $\Spc Y$ be two Hilbert spaces. Then, the linear operator $\Op T: \Spc X \to \Spc Y$ is a Parseval operator if any of the following equivalent conditions holds.
\begin{enumerate}
\item Isometry \begin{align}
\label{Eq:Isometry}
\forall x\in \Spc X: \| x\|_{\Spc X}=\|\Op T\{x\}\|_{\Spc Y}.
\end{align}
\item Preservation of inner products 
\begin{align}
\label{Eq:ParsevalCondition}
\forall x_1,x_2 \in \Spc X,\quad\langle \Op T\{x_1\},  \Op T\{x_2\}\rangle_{\Spc Y}=\langle x_1,x_2\rangle_{\Spc X}.
\end{align}
\item Pseudo-inversion via the adjoint so that $\Op T^\ast \circ \Op T=\Identity: \Spc X \to \Spc Y \to \Spc X$,
where the Hermitian adjoint $\Op T^\ast: \Spc Y \to \Spc X$ is the unique linear operator such that
\begin{align}
\label{Eq:AdjointdDef}
\forall (x,y) \in \Spc X \times \Spc Y:\quad \langle \Op T\{x\}, y \rangle_{\Spc Y}=\langle x, \Op T^\ast\{y\} \rangle_{\Spc X}.
\end{align}
\end{enumerate}
\end{proposition}
\begin{proof}

\item (i) $1 \Leftrightarrow 2$:
From the basic  properties of (real-valued) inner products and the linearity of $\Op T$, we have that
\begin{align*}
\|x_2-x_1\|^2_{\Spc X}&= \langle x_2-x_1,x_2-x_1 \rangle_{\Spc X}=\|x_2\|^2_{\Spc X}-2\langle x_1,x_2 \rangle_{\Spc X}+ \|x_1\|^2_{\Spc X}\\
\|\Op T\{x_2-x_1\}\|^2_{\Spc Y}&=\|\Op T\{x_2\}-\Op T\{x_1\}\|^2_{\Spc Y}= \|\Op T\{x_2\}\|^2_{\Spc Y}-2\langle\Op T\{x_1\},\Op T\{x_2\} \rangle_{\Spc Y}+ \|\Op T\{x_1\}\|^2_{\Spc Y}.
\end{align*}
By equating these two expressions, we readily deduce that \eqref{Eq:Isometry}
implies \eqref{Eq:ParsevalCondition}. Likewise, in the extended complex setting, we find that 
${\rm Re}(\langle x_1,x_2 \rangle_{\Spc X})={\rm Re}(\langle\Op T\{x_1\},\Op T\{x_2\} \rangle_{\Spc Y})$, which ultimately also yields \eqref{Eq:ParsevalCondition}. 
Conversely, by setting $x_1=x_2$ in \eqref{Eq:ParsevalCondition}, we directly get \eqref{Eq:Isometry}.\\[-1ex]

\item (ii) $2 \Leftrightarrow 3$: The existence and unicity of the adjoint operator $\Op T^\ast$ in \eqref{Eq:AdjointdDef} is a standard result in the theory of linear operators on Hilbert/Banach spaces. By setting $x=x_1$, $y=\Op T\{x_2\}$, and applying \eqref{Eq:AdjointdDef}, we rewrite \eqref{Eq:ParsevalCondition} as
\begin{align}
\label{Eq:Par2}
\forall x_1,x_2 \in \Spc X:
\langle \Op T\{x_1\},\Op T\{x_2\}\rangle_{\Spc Y}=\langle\Op T^\ast \Op T\{x_1\},x_2\rangle_{\Spc X}=\langle x_1, x_2\rangle_{\Spc X}.
\end{align}
Since the inner product separates all points in the Hilbert space (Hausdorff property), the right-hand side of \eqref{Eq:Par2} is equivalent to $\Op T^\ast \Op T\{ x_1\}=x_1$ for all $x_1 \in \Spc X$, which translates into $\Op T^\ast \Op T=\Op T^\ast \circ \Op T=\Identity$ on $\Spc X$.
\end{proof}
The classic example of a Parseval operator is the discrete Fourier transform $\Fourier_{d}: \ell_2(\Z^d) \to L_2(\mathbb{T}^d)$ with the Hilbertian topology specified
in Section \ref{Sec:DiscretFourier}. The fundamental property there is that the Hilbert spaces 
$\Spc X=\ell_2(\Z^d)$ and $\Spc Y=L_2(\mathbb{T}^d)$ are isomorphic with
$\Fourier_{d}^{-1}=\Fourier_{d}^{\ast}$ being a true inverse of $\Fourier_{d}$ (bijection), meaning that, in addition to Item 3 in Proposition \ref{Prop:Parseval}, we also have that $\Fourier_{d}\circ \Fourier_{d}^{\ast}=\Identity$ on $\Spc Y=L_2(\mathbb{T}^d)$ (right-inverse property).

By contrast, the Parseval convolution operators investigated in this paper will typically not be invertible from the right, the reason being that the effective range space $\widetilde{\Spc Y}=\Op T(\Spc X)$ is only a (closed) subspace of $\Spc Y$.

An important observation is that, in addition to linearity and continuity, all the operator properties in Definition \ref{Def:Parseval} are conserved through composition.

\begin{proposition}
Let $\Spc X$, $\Spc X_1$, and $\Spc X_2$ be three Hilbert spaces. If the linear operators
$\Op T_1: \Spc X \to \Spc X_1$ and $\Op T_2: \Spc X_1 \to \Spc X_2$ are both bounded (resp, bounded below with constants $A_1,A_2$ or of Parseval type), then the same holds true for
the composed operator $\Op T=\Op T_2 \circ \Op T_1: \Spc X \to \Spc X_1 \to \Spc X_2$ with 
$\|\Op T\|\le \|\Op T_1\|\; \|\Op T_2\|$ (resp., with lower bound $A=A_1A_2$).
\end{proposition}
For instance, if $\Op T_1$ and $\Op T_2$ are both bounded below, then,
for all $x\in \Spc X$,
\begin{align}
\|\Op T_2 \Op T_1\{x\}\|_{\Spc X_2}\ge A_2 \| \Op T_1\{x\}\|_{\Spc X_1} \ge A_2 A_1
\|x\|_{\Spc X}
\end{align}
with $ \Op T_1\{x\}\in \Spc X_1$.

\section{Vector-Valued LSI Operators on  $\ell^N_2(\Z^d)$}
\label{Sec:VecValued}

In this section, we shall identify and characterize the special class of linear operators that operate on discrete vector-valued signals and commute with the shift operation.

\begin{definition} 
\label{Def:LSI}
A discrete operator $\Op T_{\rm LSI}$ is {\em linear-shift-invariant} (LSI) if it is linear 
and if, for any discrete vector-valued signal $\M x[\cdot]$ in its domain and any $\bk_0\inZ^d$,
$$
\Op T_{\rm LSI}\{\M x[\cdot-\bk_0]\}=\Op T_{\rm LSI}\{\M x\}[\cdot-\bk_0].
$$
\end{definition}
We observe that the LSI property is conserved through linear combinations and composition.
Moreover, we shall see that all $\ell_2$-stable LSI operators acting on discrete vector-valued signals can be identified as (multichannel) convolution operators, as stated in Theorem \ref{Theo:L2VecBound}.

\subsection{Refresher: Scalar Convolution Operators}
To set the context, we first present a classic result on the characterization of scalar LSI operators,
together with a self-contained proof that will serve as model for subsequent derivations.

\begin{theorem}[Kernel theorem for discrete LSI operators on $\ell_2(\Z^d)$]
\label{Theo:GKerneldiscreteLSI}
For any given $h\in \ell_2(\Z^d)$, the operator $\Op T_h: x[\cdot] \mapsto (h \ast x)[\cdot]$
with $x[\cdot]\in \ell_2(\Z^d)$ and
\begin{align}
\label{Eq:ConvD}
(h \ast x)[\bk]\eqdef \langle h, x[\V k-\cdot]\rangle_{\ell_2(\Z^d)}=\sum_{\V m \in \Z^d} h[\V m]\,x[\V k-\V m],\quad \V k \inZ^d
\end{align}
is linear-shift-invariant. Moreover, $\Op T_h$ continuously maps $\ell_2(\Z^d) \to \ell_2(\Z^d)$ if and only if $\|\hat h\|_{L_\infty}=\esssup_{\bw \in [-\pi,+\pi]^d}\big|\hat h(\bw)\big|<\infty$.
Conversely, for every continuous LSI operator $\Op T_{\rm LSI}: \ell_2(\Z^d) \to \ell_2(\Z^d)$, there is one and only one $h
\in \ell_2(\Z^d)$ with $\|\hat h||_{L_\infty}=\|\Op T_{\rm LSI}\|$ such that
$\Op T_{\rm LSI}: x[\cdot] \mapsto (h \ast x)[\cdot]$ where the convolution is specified by
\eqref{Eq:ConvD}.
\end{theorem}

\begin{proof}\item {\em Direct part}. The assumption $(h,x) \in \ell_2(\Z^d)\times \ell_2(\Z^d)$ ensures that
\eqref{Eq:ConvD} is well-defined for any $\V k\inZ^d$.
The shift-invariance is then an obvious consequence of Definition \ref{Def:LSI}, as
\begin{align*}
\Op T_h\{x\}[\V k-\V k_0]&=\langle h,x[\V k-\V k_0-\cdot]\rangle_{\ell_2} \\
&=\langle h,x[(\V k-\cdot)-\V k_0]\rangle_{\ell_2}=\Op T_h\{x[\cdot-\V k_0]\}[\V k].
\end{align*}
By observing that the Fourier transform of $x[\V k-\cdot] \in \ell_2(\Z^d)$  is
$\overline{\hat x(\bw) \ee^{\jj \langle \bw, \bk\rangle}}$, we then invoke Plancherel's identity \eqref{Eq:Plancherel} to show that
 \begin{align*}
(h \ast x)[\bk]=\langle h,x[\V k-\cdot]\rangle_{\ell_2}=\int_{\mathbb{T}^d} \hat h(\bw) \hat x(\bw) \ee^{\jj \langle \bw, \bk\rangle} \;\frac{\dint \bw}{(2 \pi)^d}=\Fourier_{\rm d}^{-1}\left\{ \hat h \times \hat x\right\}[\bk],
 \end{align*}
where the identification of the inverse Fourier operator is legitimate since the boundedness of $\hat h$ implies that $\hat h \times \hat x \in L_2(\mathbb{T}^d)$.
Consequently, we are in the position where we can invoke Parseval's relation
\begin{align*}
\|h \ast x\|^2_{\ell_2}=\int_{\mathbb{T}^d} |\hat h(\bw)|^2 |\hat x(\bw)|^2 \;\frac{\dint \bw}{(2 \pi)^d}\le \|\hat h\|^2_{L_\infty}\int_{\mathbb{T}^d} |\hat x(\bw)|^2 \;\frac{\dint \bw}{(2 \pi)^d}.
\end{align*}
This yields the stability bound $\|h \ast x\|_{\ell_2}\le \|\hat h\|_{L_\infty} \|x\|_{\ell_2}$,
which implies the continuity of $\Op T_h: \ell_2(\Z^d) \to \ell_2(\Z^d)$. To show that the latter bound is sharp (``if and only if'' part of the statement), we refer to the central, more technical part of the proof of Theorem \ref{Theo:L2VecBound}. 
  
\item {\em Indirect Part}. We define the linear functional
$h: x \mapsto \langle h, x\rangle \eqdef  \Op T_{\rm LSI}\{x^\vee\}[\V 0]$. 
The continuity of $\Op T_{\rm LSI}: \ell_2(\Z^d) \to \ell_2(\Z^d)$ implies that $\Op T_{\rm LSI}\{x^\vee_i\}[\V 0]=\langle h,x_i\rangle\to 0$
for any sequence of signals $(x_i)_{i\in \N}$ in $\ell_2(\Z^d)$ that converges to 0 (or, equivalently, $x_i^\vee\to 0$). This ensures that the functional $h: x \mapsto \langle h, x\rangle$ is continuous on $\ell_2(\Z^d)$, meaning that $h \in \big(\ell_2(\Z^d)\big)'=\ell_2(\Z^d)$, which 
allows us to write that $\langle h, x\rangle=\langle h, x\rangle_{\ell_2}$. We then make use of the shift-invariance property to show that
\begin{align*}
\Op T_{\rm LSI}\{x\}[\V k]&=\Op T_{\rm LSI}\{x[\cdot+\V k]\}(\V 0)= \langle h, x[\cdot+\V k]^\vee\rangle_{\ell_2}=\langle h, x[\V k-\cdot]\rangle_{\ell_2}
\end{align*}
for any $\V k\inZ^d$, from which we also deduce that $\Op T_{\rm LSI}\{\delta\}=h$.
\end{proof}

Theorem \ref{Theo:GKerneldiscreteLSI} tells us that an LSI operator
 $\Op T_{\rm LSI}: \ell_2(\Z^d) \to \ell_2(\Z^d)$ can always be implemented as a discrete convolution with its impulse response $h=\Op T_{\rm LSI}\{\delta[\cdot]\}$.
 It also provides the Lipschitz constant of the operator, as
${\rm Lip}(\Op T_{\rm LSI})=\|\hat  h\|_{\infty}$ (supremum of its frequency response).
(We recall that, for a linear operator, the Lipschitz constant is precisely the norm of the operator.) We also note that the classic condition for stability from linear-systems theory, $h \in \ell_1(\Z^d)$, is sufficient to ensure the continuity of the operator because
$|\hat h(\bw)|\le \|h\|_{\ell_1}$. However, the latter condition is not necessary; for instance, the $\ell_2$-Lipschitz constant of an ideal lowpass pass filter is $1$ by design, while its (sinc-like)
impulse response is not included in $\ell_1(\Z^d)$.
\subsection{Multichannel Convolution Operators}
We now show that the concept carries over to vector-valued signals. To that end, we consider
a generic multichannel convolution operator that acts on an $N$-channel input signal $\M x[\cdot]\in \ell_2^N(\Z^d)$ and returns an
$M$-channel output $\M y[\cdot]$. Such an operator is characterized through
its matrix-valued impulse response $\M H[\cdot]$ with $\big[\M H[\cdot]\big]_{m,n}=h_{m,n}[\cdot]\in \ell_2(\Z^d)$ and $\M H[\bk]\inR^{M \times N}$ for any $\bk \inZ^d$. From now on, we shall denote such a convolution operator by $\Op T_{\M H}$ and refer to it as a multichannel filter. 

To benefit from the tools and theory developed for the scalar case, it is useful to express the multichannel convolution as the matrix-vector combination of a series of component-wise scalar convolutions $(h_{m,n}\ast x_n)[\cdot]$ with $(m,n) \in \{1,\dots,M\}\times \{1,\dots,N\}$.
This is written as
\begin{align}
\Op T_{\M H} 
: \M x[\cdot]\mapsto (\M H \ast \M x) [\bk] \eqdef 
\left[
\begin{array}{c}
\sum_{n=1}^N (h_{1,n} \ast x_n)[\bk]   \\
  \vdots     \\
\sum_{n=1}^N (h_{M,n} \ast x_n)[\bk] 
\end{array}
\right],
\label{Eq:ConvVec}
\end{align}
where \begin{align}
\M H[\cdot]=\left[
\begin{array}{ccc}
h_{1,1}[\cdot] & \cdots & h_{1,N}[\cdot]\\
\vdots &\ddots &\vdots \\
h_{M,1}[\cdot] & \cdots & h_{M,N}[\cdot]\end{array}
\right]=
\left[\M h_1[\cdot] \ \cdots \ \M h_N[\cdot]\right]
\end{align}
with the $n$th column of the impulse response being identified as
$$\M h_n[\cdot]=\left[\begin{array}{c}
h_{1,n}[\cdot]   \\
  \vdots      \\
h_{M,n}[\cdot] 
\end{array}\right]
=\Op T_{\M H}\{\M e_n\delta[\cdot]\}.$$
We also note that the convolution in \eqref{Eq:ConvVec} has an explicit representation, given by \eqref{Eq:ConvmultiD}, which is the matrix-vector counterpart of the scalar formula \eqref{Eq:ConvD}.

As in the scalar scenario, the multichannel convolution can be implemented by a multiplication in the Fourier domain, with the frequency response of the filter now having the form of a matrix.
Specifically, for any $\M x[\cdot]\in \ell_2^N(\Z^d)$ with vector-valued Fourier transform
$\widehat{\M x}=\Fourier_{\rm d}\{\M x[\cdot]\} \in L_2^N(\mathbb{T}^d)$, we have that
%
\begin{align}
\Fourier_d\big\{ (\M H \ast \M x)[\cdot]\big\}(\bw)&=\widehat{\M H}(\bw) \, \widehat{\M x}(\bw)=\left[\widehat{\M h}_1(\bw) \ \cdots \ \widehat{\M h}_N(\bw)\right]\, \widehat{\M x}(\bw)
 \nonumber \\
&=
\left[
\begin{array}{ccc}
\hat h_{1,1}(\bw) & \cdots & \hat h_{1,N}(\bw)\\
\vdots&\ddots&\vdots \\
\hat h_{M,1}(\bw) & \cdots & \hat h_{M,N}(\bw)\end{array}
\right]
\, \left[
\begin{array}{c}
\hat x_1(\bw)\\
\vdots\\
\hat x_N(\bw)
\end{array}
\right],\end{align}
where the matrix-valued function $\widehat{\M H}: [-\pi,\pi]^d \to \C^{N \times M}$, with
$[\widehat{\M H}]_{m,n}=\hat h_{m,n}=\Fourier_{\rm d}\{h_{m,n}\}$, is the component-by-component Fourier transform of
the matrix filter $\M H[\cdot]$.


\subsection{Kernel Theorem for Multichannel LSI Operators}
The matrix-vector convolution specified by 
\eqref{Eq:ConvVec} 
is well-defined for any $\M x[\cdot
] \in \ell_2^N(\Z^d)$
under the assumption that $\M H[\cdot] \in \ell_2(\Z^d)^{M \times N}$. Yet, we need to be a bit more selective to ensure that the operator is (Lipschitz-) continuous with respect to the $\ell_2$-norm. We show in Theorem \ref{Theo:L2VecBound}
that there is an equivalence between continuous multi-channel LSI operators and bounded multichannel filters (convolution operators), while we also give an explicit formula for the norm of the operator. As one may expect, the Schwartz kernel of the LSI operator is the matrix-valued impulse response of the multichannel filter.

\begin{theorem} [Kernel theorem for LSI operators $\ell^N_2(\Z^d)\to \ell^M_2(\Z^d)$]
\label{Theo:L2VecBound}
For any given $\M H[\cdot] \in \ell_2(\Z^d)^{M \times N}$, the convolution operator $\Op T_{\M H}: \M x[\cdot] \mapsto (\M H \ast \M x)[\cdot]$
with $N$-vector-valued input $\M x[\cdot]\in \ell^N_2(\Z^d)$ and $M$-vector-valued output
\begin{align}
\label{Eq:ConvmultiD}
(\M H \ast \M x) [\bk] =\sum_{\V \ell \in \Z^d} \M H[\V \ell]\M x[\V k-\V \ell],\quad \V k \inZ^d
\end{align}
is linear-shift-invariant and characterized by its matrix-valued frequency response $\Fourier_{\rm d}\{\M H[\cdot]\}=\widehat{\M H}(\cdot) \in
L_2(\mathbb{T}^d)^{M \times N}$.
%
Moreover, $\Op T_{\M H}$ continuously maps
$\ell^N_2(\Z^d) \to \ell^M_2(\Z^d)$ if and only if 
\begin{align}
\label{Eq:LH}
\|\Op T_{\M H}\|=\sigma_{\sup,\M H}
=\esssup_{\bw \in [-\pi,\pi]^d} \sigma_{\max}\big(\widehat{\M H}(\bw)\big)
<\infty,
\end{align}
where $\sigma_{\max}\big(\widehat{\M H}(\bw)\big)$ with $\bw$ fixed is the maximal singular value of the matrix
$\widehat{\M H}(\bw) \in \C^{M \times N}$.

Conversely, for every continuous LSI operator $\Op T_{\rm LSI}: \ell_2^N(\Z^d) \toC \ell^M_2(\Z^d)$, there is one and only one $\M H[\cdot]
 \in \ell_2(\Z^d)^{M \times N}$ (the matrix-valued impulse response of $\Op T_{\rm LSI}$) such that
$\Op T_{\rm LSI}=\Op T_{\M H}: \M x[\cdot] \mapsto (\M H \ast \M x)[\cdot]$
and $\|\Op T_{\rm LSI}\|_{\ell^N_2 \to \ell^M_2}=\sigma_{\sup,\M H} < \infty$.

\end{theorem}
\begin{proof} 
\item {\em Direct Part}.  The $m$th entry of $(\M H \ast \M x) [\bk]$
can be identified as $\big[(\M H \ast \M x) [\bk]\big]_m=\langle \M g_m,\M x[\V k - \cdot] \rangle_{\ell_2^N(\Z^d)}$ with 
$\M g_m[\cdot]=(h_{m,1}[\cdot],\dots,h_{m,N}[\cdot]) \in \ell_2^N(\Z^d)$ being the $m$th row of the matrix-valued impulse response 
$\M H[\cdot]$. 
The LSI property (see Definition \ref{Def:LSI}) then 
follows from the observation that
\begin{align*}
\big[(\M H \ast \M x) [\bk-\bk_0]\big]_m&=\langle \M g_m,\M x[(\V k -\V k_0)-\cdot)] \rangle_{\ell_2^N(\Z^d)}\\
&=\langle \M g_m,\M x[(\V k -\cdot)-\V k_0] \rangle_{\ell_2^N(\Z^d)}=\big[(\M H \ast \M x[\cdot-\bk_0]) [\bk]\big]_m,
\end{align*}
for $m=1,\dots,M$ and any $\V k, \V k_0 \in \Z^d$.

The Fourier-domain equivalent of the hypothesis $\M x \in \ell^N_2(\Z^d)$
(resp. $\M H[\cdot] \in \ell_2(\Z^d)^{M \times N}$) is $\widehat{\M x} \in L^N_2(\mathbb{T}^d)$
(resp., $\widehat{\M H}(\cdot) \in L_2(\mathbb{T}^d)^{M \times N}$). The key for this equivalence is the vector-valued version of Parseval's identity given by
\begin{align*}
\|\M x\|^2_{\ell^N_2(\Z^d)}=\sum_{n=1}^N \|x_n\|_{\ell_2}^2&= \sum_{n=1}^N  \int_{\mathbb{T}^d}  |\hat x_n(\bw)|^2 \frac{\dint \bw}{(2 \pi)^d}
= \int_{\mathbb{T}^d} 
\|\widehat{\M x}(\bw)\|_2^2 \frac{\dint \bw}{(2 \pi)^d}= \|\widehat{\M x}\|^2_{L^N_2(\mathbb{T}^d)}.
\end{align*}
Likewise, under the assumption that $\widehat{\M H}(\cdot)\widehat{\M x}(\cdot) \in L_2^M(\mathbb{T}^d)$, we can evaluate the $\ell_2$-norm of the convolved signal as
\begin{align}
\label{Eq:FourierConvNorm}
\|\M H \ast \M x\|_{\ell^M_2(\Z^d)}&= 
\left(\int_{\mathbb{T}^d} \|\widehat{\M H}(\bw) \widehat{\M x}(\bw)\|_2^2 \frac{\dint \bw}{(2 \pi)^d}\right)^\frac{1}{2}=\|\widehat{\M H}\widehat{\M x}\|_{L^M_2(\mathbb{T}^d)},
\end{align}
where we are relying on the property that the convolution corresponds to
a pointwise multiplication in the Fourier domain.\\[-1ex]

\item {\em Norm of the Operator
}.
Implicit in the specification of $\sigma_{\sup,\M H}$ in \eqref{Eq:LH} is the requirement that the matrix-valued frequency response
$\widehat{\M H}(\cdot): \mathbb{T}^d \to \C^{M \times N}$ 
be measurable and bounded almost everywhere.
This means that $\widehat{\M H}(\bw)$ with $\bw$ fixed is a well-defined matrix in $\C^{M \times N}$ for almost any $\bw \in \mathbb{T}^d$. In that case, we can specify its maximal 
singular values by
\begin{align*}
\sigma_{\max}\big(\widehat{\M H}(\bw)\big)=\sup_{\M u \in \C^N \backslash \{ \V 0\}} \frac{\|\widehat{\M H}(\bw) \M u\|_2}{\ \|\M u\|_2}
\end{align*}
Consequently, for any $\widehat{\M x}(\cdot) \in L^N_2(\mathbb{T}^d)$, we have that
\begin{align*}
\|\widehat{\M H}(\bw) \widehat{\M x}(\bw)\|_2 \le 
\|\widehat{\M x}(\bw)\|_2 \cdot \sigma_{\max}\big(\widehat{\M H}(\bw)\big) \le
\|\widehat{\M x}(\bw)\|_2 \cdot \sigma_{\sup,\M H} 
\end{align*}
for almost any $\bw \in \mathbb{T}^d$. This implies that
\begin{align}
\label{Eq:BoundedMatConv}
\|\widehat{\M H}\widehat{\M x}\|_{L^M_2(\mathbb{T}^d)}& =\left(\int_{\mathbb{T}^d} \| \widehat{\M H}(\bw) \widehat{\M x}(\bw)\|_2^2 \frac{\dint \bw}{(2 \pi)^d}\right)^\frac{1}{2}\nonumber\\&
\le  \sigma_{\sup,\M H}  \left(\int_{\mathbb{T}^d}\|\widehat{\M x}(\bw)\|_2^2  \frac{\dint \bw}{(2 \pi)^d}\right)^\frac{1}{2}= \sigma_{\sup,\M H} \cdot  \|\widehat{\M x}\|_{L^N_2(\mathbb{T}^d)}
\end{align}
which, due to the Fourier isometry, yields the upper bound $\|\Op T_{\M H}\|\le \sigma_{\sup,\M H}$.

Likewise, \eqref{Eq:FourierConvNorm} implies that $\|\Op T_{\M H}\|=\|\widehat{\M H}\|$, 
which is the norm of the pointwise multiplication operator $\widehat{\M x} \mapsto
\widehat{\M H}\widehat{\M x}$ and is equal to $\|\sigma_{\max}\big(\widehat{\M H}(\cdot)\big)\|_{L_\infty}$.
Indeed, for any $\widehat{\M x}(\cdot) \in \Spc S(\mathbb{T}^d)^N \subset L_2^N(\mathbb{T}^d)$, the boundedness of $\widehat{\M H}: L_2^N(\mathbb{T}^d) \to L_2^M(\mathbb{T}^d)$ 
implies that
\begin{align*} 
 \int_{[-\pi,\pi]^d} \widehat{\M x}^\HTop(\bw) \widehat{\M H}^\HTop(\bw) \widehat{\M H}(\bw) \widehat{\M x}(\bw) \frac{\dint \bw}{(2 \pi)^d} \le \|\widehat{\M H}\|^2 \int_{[-\pi,\pi]^d} \|\widehat{\M x}(\bw)\|_2^2 \frac{\dint \bw}{(2 \pi)^d},
\end{align*}
which is equivalent to
\begin{align*}
 \int_{[-\pi,\pi]^d} \widehat{\M x}^\HTop(\bw) \left(\|\widehat{\M H}\|^2 \M I_N-\widehat{\M H}^\HTop(\bw) \widehat{\M H}(\bw)\right) \widehat{\M x}(\bw) \frac{\dint \bw}{(2 \pi)^d} \ge 0.\end{align*}
This relation
implies that
the Hermitian-symmetric matrix $\left(\|\widehat{\M H}\|^2  \M I_N-\widehat{\M H}^\HTop(\bw) \widehat{\M H}(\bw)\right)$
is nonnegative-definite for almost
any $\bw \in \mathbb{T}^d$. On the side of the eigenvalues, this translates into
$$\|\widehat{\M H}\|^2-\lambda_{\max}\left(\widehat{\M H}^\HTop(\bw) \widehat{\M H}(\bw)\right)=\|\widehat{\M H}\|^2-\sigma^2_{\max}\left(\widehat{\M H}(\bw)\right)
\ge 0 \  \ a.e.$$ leading to
$\|\sigma_{\max}\left(\widehat{\M H}(\cdot)\right)\|_{L_\infty}=\sigma_{\sup,\M H}\le \|\widehat{\M H}\|=\|\Op T_{\M H}\|$. Since we already know that $\|\Op T_{\M H}\| \le \sigma_{\sup,\M H}$, we deduce that $\|\Op T_{\M H}\|= \sigma_{\sup,\M H}$.
\\[-1ex]

\item {\em Indirect Part}. We define the linear functionals
$\M g_m: \M x[\cdot] \mapsto \langle \M g_m, \M x\rangle =  \Big[\Op T_{\rm LSI}\{\M x^\vee\}[\V 0]\Big]_m$ with $m\in\{1,\dots,M\}$.
The continuity of $\Op T_{\rm LSI}: \ell^N_2(\Z^d) \to \ell^M_2(\Z^d)$ implies that $\Big[\Op T_{\rm LSI}\{\M x^\vee_i\}[\V 0]\Big]_m=\langle \M g_m,\M x_i\rangle\to 0$
for any converging sequence $\M x_i[\cdot] \to \V 0$ (or, equivalently, $\M x_i^\vee[\cdot]\to \M 0$) in $\ell^N_2(\Z^d)$. This ensures that the functional $\M g_m: \M x \mapsto \langle \M g_m, \M x\rangle$ is continuous on $\ell^N_2(\Z^d)$, which is equivalent to $\M g_m=(g_{n,m}[\V k])_{(n,\V k) \in \{1,\dots,N\}\times \Z^d} \in \big(\ell^N_2(\Z^d)\big)'=\ell^N_2(\Z^d)$. This then allows us to write $\langle \M g_m, \M x\rangle=\langle \M g_m, \M x\rangle_{\ell^N_2}$ for all $m\in\{1,\dots,M\}$. We then make use of the shift-invariance property to show that
\begin{align*}
\Op T_{\rm LSI}\{\M x\}[\V k]=\Op T_{\rm LSI}\{\M x[\cdot+\V k]\}[\V 0]&=
\begin{pmatrix}\langle \M g_1, \M x^\vee[\cdot+\V k]\rangle_{\ell^N_2}\\
\vdots\\
 \langle \M g_M, \M x^\vee[\cdot+\V k]\rangle_{\ell^N_2}\end{pmatrix} \\
 &=\begin{pmatrix} \sum_{n=1}^N \sum_{\V m \in\Z^d} g_{n,1}[\V m] x_n[\V k-\V m]\\
\vdots\\
\sum_{n=1}^N \sum_{\V m \in\Z^d} g_{n,M}[\V m] x_n[\V k-\V m]\end{pmatrix}\\
&=\begin{pmatrix} \sum_{n=1}^N (g_{n,1}\ast  x_n)[\V k]\\
\vdots\\
\sum_{n=1}^N (g_{n,M}\ast  x_n)[\V k]
\end{pmatrix}
=(\M H \ast \M x) [\bk]
\end{align*}
for any $\V k\inZ^d$, from which we deduce that $\Op T_{\rm LSI}=\Op T_{\M H}$ with
matrix-valued impulse response $\M H[\cdot]$ 
whose entries are $h_{m,n}[\V k]=\big[\M g_m[\V k]\big]_n=g_{n,m}[\V k]$ with $g_{n,m}[\cdot] \in \ell_2(\Z^d)$.
\end{proof}

An immediate consequence is that the composition of the two continuous LSI operators $\Op T_{\M H_1}: \ell_2^N(\Z^d) \to \ell_2^{N_2}(\Z^d)$
and $\Op T_{\M H_2}: \ell_2^{N_2}(\Z^d) \to \ell_2^M(\Z^d)$ yields a stable multi-filter $\Op T_{\M H}=\Op T_{\M H_2 \ast \M H_1}: \ell_2^N(\Z^d) \to \ell_2^{M}(\Z^d)$ with $\|\Op T_{\M H}\|\le\sigma_{\sup,\M H_1}\,\sigma_{\sup,\M H_2}$.
The frequency response of the composed filter is the product $\widehat{ \M H}(\bw)=\widehat{ \M H}_1(\bw)\widehat{ \M H}_2(\bw)$ of the individual responses, as expected. On the side of the impulse response, this translates into the matrix-to-matrix convolution 
\begin{align}
\label{Eq:LSICompose}
(\M H_2 \ast \M H_1) [\bk] \eqdef \sum_{\V m\in \Z^d} \M H_2[\V m]\M H_1[\V k-\V m],\quad \V k \inZ^d,
\end{align}
which is the matrix counterpart of \eqref{Eq:ConvD}. Beside the fact that the inner dimension $(N_2$) of the matrices must  match, an important difference with the scalar setting is that matrix convolutions are generally not commutative.

\subsection{Parseval Filterbanks}
We now proceed with the characterization of the complete family of Parseval LSI operators from $\ell^N_2(\Z^d)\to \ell^M_2(\Z^d)$.
We know from Theorem \ref{Theo:L2VecBound} that these are necessarily filterbanks of the form $\Op T_{\M H}: \M x[\cdot] \mapsto 
(\M H \ast \M x)[\cdot]$, which can also be specified by their
matrix-valued frequency response $\widehat{\M H}(\cdot)$. Moreover, Proposition \ref{Prop:Parseval}
tells us that the Parseval condition is equivalent to $\Op T_{\M H}^\ast\circ \Op T_{\M H}=\Identity$.

Consequently, the only remaining part is to identify the adjoint operator $\Op T_{\M H}^\ast:\ell^M_2(\Z^d)\to \ell^N_2(\Z^d)$, which is done through the manipulation
\begin{align*}
\forall (\M x, \M y) \in \ell^N_2(\Z^d)\times \ell^M_2(\Z^d):\hspace*{7cm}\nonumber\\
\langle \M y , (\M H \ast \M x)[\cdot]\rangle_{\ell_2^M(\Z^d)}=\langle \widehat{\M y} ,\widehat{ \M H} \widehat{\M x}\rangle_{L_2^M(\mathbb{T}^d)}=\langle \widehat{ \M H}^\HTop\widehat{\M y} , \widehat{\M x}\rangle_{L_2^N(\mathbb{T}^d)}=\langle (\M H^{\Tr\vee} \ast \M y)[\cdot] , \M x)\rangle_{\ell_2^N(\Z^d)},
\end{align*}
where we used the Fourier-Plancherel isometry, a pointwise Hermitian transposition to move the frequency-response matrix on the other side of the inner product, and the property that a complex conjugation of the frequency response translates into the flipping of the impulse response. Based on \eqref{Eq:AdjointdDef}, we can then identify $\Op T^\ast_{\M H}: \M y \mapsto (\M H^{\Tr\vee} \ast \M y)[\cdot]$. This shows that the adjoint of $\Op T_\M H$ is  
 the convolution operator whose matrix impulse response is
$\M H^{\Tr\vee}[\cdot]$ (the flipped and transposed version of $\M H[\cdot]$) and whose frequency response is $\Fourier_{\rm d}\{\M H^{\Tr\vee}[\cdot]\}=\widehat{\M H}^\HTop(\cdot)$.

\begin{proposition}[Characterization of Parseval-LSI operators] 
\label{Theo:PasevalFilters}
A
linear operator $\Op T :\ell^N_2(\Z^d)\to \ell^M_2(\Z^d)$ with $M\ge N$ is
LSI and energy-preserving (Parseval) if and only if it can be represented as a multichannel filterbank
$\Op T=\Op T_{\M H}: \M x[\cdot] \mapsto (\M H \ast \M x)[\cdot]$ whose matrix-valued  impulse response  $\M H[\bk]\in \R^{M \times N}$ with $\bk$ ranging over $\Z^d$ has any of the following equivalent properties.
\begin{enumerate}
\item Invertibility by flip-transposition: $$(\M H^{\Tr\vee} \ast \M H)[\cdot]= \M I_N \delta[\cdot],$$ which is equivalent to $\Op T^\ast_{\M H}\circ \Op T_{\M H}=\Identity$ on $\ell^N_2(\Z^d)$.
\item Paraunitary frequency response: $$\widehat{\M H}^\HTop(\bw) \widehat{\M H}(\bw)=\M I_N \mbox{ for all }\bw \in \mathbb{T}^d,$$ where
$\widehat{\M H}=\Fourier_{\rm d}\{\M H[\cdot]\}\in L_2(\mathbb{T}^d)^{M \times N}$ is the discrete Fourier transform of $\M H[\cdot]$.
\item Preservation of inner products:
$$\forall \M x, \M y \in \ell^N_2(\Z^d):\quad
\langle \M x, \M y \rangle_{\ell^N_2(\Z^d)}=\langle(\M H \ast \M x)[\cdot], (\M H \ast \M y)[\cdot]\} \rangle_{\ell^M_2(\Z^d)}.
$$
\end{enumerate}
\end{proposition}
We also note that the LSI-Parseval property implies that
$\|\M H[\cdot]\|_{\ell_2^{M \times N}(\Z^d)}=N$ and $\|\Op T_{\M H}\|=1$, 
although those conditions are obviously not sufficient.

While Item 1 suggests that the adjoint $\Op T^\ast_{\M H}=\Op T_{\M H^{\Tr\vee}}:
\ell_2^M(\Z^d) \to \ell_2^N(\Z^d) $ acts as the inverse of
$\Op T_{\M H}$, this is only true for signals $\M y[\cdot] \in \Op T_{\M H}\big(\ell_2^N(\Z^d)\big) \subset \ell^M_2(\Z^d)$ that are in the range of the operator. In other words, $\Op T^\ast_{\M H}$ is only a left inverse of $\Op T^\ast_{\M H}$, while is fails to be a right inverse in general, unless $M=N$. This is denoted by $\Op T^\ast_{\M H}=\Op T_{\M H}^+$ (generalized inverse).

While the $M$-to-$N$ filter $\Op T^\ast_{\M H}$ is generally not a Parseval filter, it is $1$-Lipschitz (since $\|\Op T^\ast_{\M H}\|=\|\Op T_{\M H}\|=1$) with its Gram
 operator
$(\Op T^{\ast\ast}_\M H \circ \Op T^\ast_\M H)=(\Op T_\M H \circ \Op T^\ast_\M H): \ell^M_2(\Z^d) \to \ell^M_2(\Z^d)$ being the orthogonal projector on the range of $\Op T_\M H$, rather than the identity.
%
%
%
Correspondingly, from the properties of the singular value decomposition (SVD), we can infer that $\widehat{\M H}^\HTop(\bw)$ with $\bw$ fixed has the same nonzero singular values as $\widehat{\M H}(\bw)$ ($N$ singular values equal to one) and that these are complemented with $(M-N)$ additional zeros to make up for the fact that $M>N$.

The filterbanks used in convolutional neural network are generally FIR, meaning that their matrix impulse response is finitely supported. This is the reason why the reminder of the chapter is devoted to the investigation of FIR-Parseval convolution operators.
To set the stage, we start with the single-channel case $N=M=1$, which has the fewest degrees of freedom.
\begin{proposition}
\label{Prop:Shifts}
The real-valued LSI operator $\Op T_h: \ell_2(\Z^d) \to \ell_2(\Z^d)$ is FIR-Parseval if and 
only if $h=\pm \delta[\cdot-\bk_0]$  for some $\bk_0 \in \Z^d$. Equivalently, $\Op T_h=\pm \Op S^{\bk_0}$ where 
$\Op S^{\bk_0}: x[\cdot]\mapsto  x[\cdot-\bk_0]$.
\end{proposition}
\begin{proof} From 
Proposition \ref{Theo:PasevalFilters}, we know that the LSI-Parseval property is equivalent to $(h^\vee \ast h)[\bk]=\sum_{\V m \in \Z^d}h[-\V m]h[\bk -\V m]=\sum_{\V m \in \Z^d}h[\V m]h[\V m+\bk ]=\delta[\bk]$, which is obviously met
for $h=\pm \delta[\cdot-\bk_0]$.  Now, if ${\rm supp}(h)=\{\V m \in \Z^d:  h[\V m]\ne 0\}$ is finite and includes at least two distinct points, then there always exists some critical offset $\bk_0\ne 0$ such that
${\rm supp}(h) \cap {\rm supp}(h[\cdot +\bk_0])=\{\V m_0\}$; in other words, such that the intersection of the support and its shifted version by $-\bk_0$ consists of a single point.
Consequently, $\sum_{\V m \in \Z^d}h[\V m]h[\V m+\bk_0]=h[\V m_0]h[\V m_0+\bk_0]\ne 0$, which is incompatible with the definition of the Kronecker delta.
\end{proof}

Proposition \ref{Prop:Shifts} identifies the shift operators as fundamental LSI-Parseval elements, but the family is actually larger if we relax the FIR condition.
The frequency-domain condition for Parseval is $|\hat h(\bw)|=1$, which translates into the filter being all-pass. Beside any power of the shift operator, a classic example for $d=1$ is $\hat h(\omega)=\frac{\ee^{-\jj \omega}-\overline{z_0}}{1-z_0\ee^{-\jj \omega}}$, with the caveat that the impulse response of the latter is infinitely supported.

 \section{Parametrization of Parseval Filterbanks}
 \label{Sec:Parameterization}
 While the design options for (univariate) FIR Parseval filters are fairly limited (see Proposition \ref{Prop:Shifts}), we now show that
the possibilities open up considerably in the multichannel setting. This is good news for applications.

Our approach to construct trainable FIR Parseval filterbanks
is based on the definition of basic $1$-to-$N$, $N$-to-$N$, and $N$-to-$(pN)$
Parseval filters that can then be chained, 
in the spirit of neural networks, to produce more complex structures. Specifically, let $\Op T_{\M H_i}: \ell_2^{N_{i}}(\Z^d) \to \ell_2^{N_{i+1}}(\Z^d), i=1,\dots,I$
be a series of Parseval filters with $N_1=N\le N_i \le \cdots \le N_{I+1}=M$ and
$\Op T^\ast_{\M H_i}\circ \Op T_{\M H_i}= \Identity$ on $\ell_2^{N_{i}}(\Z^d)$.
Because the LSI and Parseval properties are preserved through composition,
one immediately deduces that the composed operator
\begin{align}
\Op T_{\M H}=\Op T_{\M H_I}\circ \cdots \circ\Op T_{\M H_1}: \ell_2^{N}(\Z^d) \to \ell_2^{M}(\Z^d)
\end{align}
is Parseval-LSI with 
impulse response $\M H[\cdot]=(\M H_I \ast \cdots \ast \M H_2 \ast \M H_1)[\cdot] \in \ell_2(\Z^d)^{M \times N}$. This filter is invertible from the left with its generalized inverse being
\begin{align}
\Op T^+_{\M H}=\Op T^\ast_{\M H}=\Op T^\ast_{\M H_1}\circ \cdots \circ\Op T^\ast_{\M H_I}: \ell_2^{M}(\Z^d) \to \ell_2^{N}(\Z^d),
\end{align}
which means that the inverse filtering can be achieved via a simple flow-graph transposition of the original filter architecture. 

Thus, our design concept is to rely on simple elementary modules, each being parameterized by an orthogonal matrix $\M U_i\in \C^{M \times M}$ where $M$ is typically the number of output channels.
The list of our primary modules is summarized in Table \ref{Tab:ParsevalFilters} . Additional detailed descriptions and explanations are given in the  remainder of this section.

\begin{table}\label{tab:activation}
\begin{tabular}{p{210pt} p{120pt}}
\hline \hline &\\[-1.8ex]\quad LSI-Parseval Operators & Impulse Response \\[1ex]
\hline\\[-2ex]
Patch descriptor $\mathbb{K}_M=\{\bk_1,\dots,\bk_M\}$\\
\quad ${\rm Patch}_{\mathbb{K}_M}: \ell_2^N(\Z^d) \to \ell_2^{M \times N}(\Z^d)$ & $\frac{1}{\sqrt{M}}\begin{pmatrix}\M I_N  \delta[\cdot-\bk_1]\\
\vdots \\
\M I_N  \delta[\cdot-\bk_M]\end{pmatrix}$\\[1ex]
Unitary matrix $\M U=[\M u_1 \dots \M u_N] \in \C^{N \times N}$\\[1ex]
\quad ${\rm Mult}_{\M U}: \ell_2^N(\Z^d) \to \ell_2^N(\Z^d)$ & $\M U \delta[\cdot]$\\[1ex]
\quad $\Op T_\M h= \M U\, {\rm Patch}_{\mathbb{K}_N}: \ell_2(\Z^d) \to \ell_2^N(\Z^d)$ & $\displaystyle \M h[\cdot]=\frac{1}{\sqrt{N}}\sum_{n=1}^N \M u_n\delta[\cdot-\bk_n]$\\[1ex]
Large unitary matrix $\M U=[\M U_1 \cdots \M U_p] \in \C^{p N \times p N}$\\
\quad $\M U\, {\rm Patch}_{\mathbb{K}_N}: \ell^N_2(\Z^d) \to \ell_2^{pN}(\Z^d)$ & $\displaystyle \frac{1}{\sqrt{p}}\sum_{n=1}^p \M U_n\delta[\cdot-\bk_n]$\\[1ex]

Generalized shift with $\V K=(\bk_1,\dots,\bk_N) \in \Z^{d \times N}$\\
\quad $\Op S^{\V K}: \ell_2^N(\Z^d) \to \ell^N_2(\Z^d)$ & ${\rm diag}(\delta[\cdot -\bk_1],\dots,\delta[\cdot -\bk_N])$
\\[1ex]
Frame matrix $\M A \in \C^{M \times N}$ s.t. $\M A^\HTop\M A=\M I_N$\\
\quad$\M A\Op S^{\V K}: \ell^N_2(\Z^d) \to \ell_2^{ M}(\Z^d)$ & $\displaystyle\sum_{n=1}^N \M a_n\M e_n^\HTop \delta[\cdot- \bk_n]$\\[1ex]
Unitary matrices  $\M U, \M V \in \C^{N \times N}$\\[1ex]
\quad$\M U\Op S^{\V K}: \ell^N_2(\Z^d) \to \ell_2^{ N}(\Z^d)$ & $\displaystyle\sum_{n=1}^N \M u_n \M e_n^\HTop \delta[\cdot- \bk_n]$\\[1ex]
\quad$\M U\Op S^{\V K}\M V^\HTop: \ell^N_2(\Z^d) \to \ell_2^{ N}(\Z^d)$ & $\displaystyle\sum_{n=1}^N \M u_n \M v_n^\HTop \delta[\cdot- \bk_n]$\\[3ex]
Rank-$k$ projector  $\M P_k=\M U_k\M U_k^\HTop \in \R^{N \times N}$ with $\M P^2_k=\M P_k=\M P_k^\Tr$\\[1ex]
\quad$\Op P_{\M P_k,n}: \ell^N_2(\Z^d) \to \ell_2^{ N}(\Z^d)$ & $(\M I_N - \M P_k) \delta[\cdot] + \M P_k\delta[\cdot- \M e_n]$\\[2ex]
Householder element with $\M u \in \C^N$ s.t.\ $\|\M u\|_2=1$  \\[1ex]
\quad$\Op H_{\M u,n}: \ell^N_2(\Z^d) \to \ell_2^{ N}(\Z^d)$ & $(\M I_N - \M u\M u^\HTop) \delta[\cdot] + \M u\M u^\HTop \delta[\cdot- \M e_n]$\\[1ex]
\hline\hline
\end{tabular}
\caption{Elementary parametric Parseval multi-filters. There, most filters are parameterized by a unitary matrix/frame and
a list of neighborhood indices $\bk_1,\dots,\bk_N$ (not necessarily distinct). The vector $\M e_n$ with $[\M e_n]_m=\delta_{n-m}$ is the $n$th element of a canonical basis. 
\label{Tab:ParsevalFilters}}
\end{table}

\subsection{Normalized patch operator}
Our first tool is a simple 
mechanism to augment the number of output channels of the filterbank.
It involves a patch of size $M$ specified by a list $\mathbb{K}_M=\{\bk_1, \cdots \bk_M\}$ of indices, which will thereafter be used to
describe the support of filters acting on each feature channel.
Our normalized patch operator extracts
 the signal values within the patch in running fashion as \begin{align}
{\rm Patch}_{\mathbb{K}_M}: \M x[\cdot] \mapsto \frac{1}{\sqrt{M}}\begin{pmatrix}\M x[\cdot-\bk_1]\\
\vdots\\
\M x[\cdot-\bk_M]
\end{pmatrix}.
\end{align}
One easily checks that ${\rm Patch}_{\mathbb{K}_M}: \ell_2^{N}(\Z^d) \to \ell_2^{M\times N}(\Z^d)$ is LSI and Parseval, because the $\ell_2$-norm is invariant to a shift and conserved in each of the output components---the very reason why the output is normalized by $\sqrt{M}$. Its adjoint, ${\rm Patch}^\ast_{\mathbb{K}_M}: \ell_2^{M \times N}(\Z^d) \to  \ell_2^{N}(\Z^d)$, is the signal recomposition operator 
\begin{align}
{\rm Patch}^\ast_{\mathbb{K}_M}:\begin{pmatrix}\M y_1[\cdot]\\
\vdots\\
\M y_M[\cdot]
\end{pmatrix} \mapsto \frac{1}{\sqrt{M}}\sum_{m=1}^M \M y_m[\cdot+\bk_m],
\end{align}
where the $\M y_m[\cdot]$ are $N$-vector-valued signals. 
The fundamental property for this construction is
${\rm Patch}^\ast_{\mathbb{K}_M}\circ {\rm Patch}_{\mathbb{K}_M}=\Identity$ on $\ell^N_2(\Z^d)$, as direct consequence of the isometric nature of the operator. 

For $N=1$, the impulse response of  ${\rm Patch}^\ast_{\mathbb{K}_M}$  is 
$\tfrac{1}{\sqrt{M}}\begin{bmatrix}\delta(\cdot+\bk_1)&\cdots&\delta(\cdot+\bk_M)\end{bmatrix}$ whose vector-valued Fourier transform is
$\tfrac{1}{\sqrt{M}}
\begin{bmatrix}\ee^{\jj \langle\bw, \bk_1\rangle} & \cdots &\ee^{\jj \langle\bw, \bk_M\rangle}\end{bmatrix}$. The paraunitary nature of this system is revealed in the basic relation
\begin{align}
\tfrac{1}{\sqrt{M}} \begin{bmatrix}\ee^{\jj \langle\bw, \bk_1\rangle}&
\cdots &
\ee^{\jj \langle\bw, \bk_M\rangle}
\end{bmatrix} \tfrac{1}{\sqrt{M}} \begin{bmatrix}\ee^{-\jj \langle\bw, \bk_1\rangle}\\
\vdots\\[1ex]
\ee^{-\jj \langle\bw, \bk_M\rangle} 
\end{bmatrix}
= \frac{\sum_{m=1}^M |\ee^{\jj \langle\bw, \bk_m\rangle}|^2 }{M}=1,
\end{align}
which holds for any choice of the $\bk_m$.

\subsection{Parametric $1$-to-$N$ Parseval Module}
\label{Sec:Parselva1toN}
The necessary and sufficient condition for a $1$-to-$N$ operator $\Op T_\M h: \ell_2(\Z^d) \to \ell_2^N(\Z^d)$ to have the Parseval property is
$$
(\M h^{\Tr \vee}\ast \M h)[\cdot]=\sum_{n=1}^N (h_n^\vee \ast h_n)[\cdot]=\delta[\cdot],
$$
which, once stated in in the frequency domain, 
 is 
\begin{align}
 \label{Eq:Parseval3}
\forall \bw \in \mathbb{T}^d:\quad \|\widehat {\M h}(\bw)\|^2= \sum_{n=1}^N |\hat h_n(\V \omega)|^2=1.
 \end{align}
This indicates that the frequency responses of the component filters $h_n$ should be power complementary. This is a standard requirement in wavelet theory and the construction of tight frames, which has been
the basis for various parametrizations \cite{Vetterli:1995,Strang:1996}.

What we propose here is a simple matrix-based construction of such filters
with the support of each filter also being of size $N$.
The filtering window, which is common to all channels and assimilated to a patch, is specified by
the index set $\mathbb{K}_N=\{\bk_1,\dots,\bk_N\}$. These indices are usually chosen to be contiguous and centered around the origin. For instance, $\mathbb{K}_3=\{-1,0,1\}$ specifies centered filters of size $3$ in dimension $d=1$.
Given some orthogonal matrix $\M U=\begin{bmatrix}\M u_1& \cdots& \M u_N\end{bmatrix} \in \R^{N\times N}$,
our basic parametric $1$-to-$N$ filtering operator is then given by 
\begin{align}
\label{Eq:1toNelement}
\Op T_{\M h} = {\rm Mult}_{\M U} \circ {\rm Patch}_{\mathbb{K}_N}: \ell_2(\Z^d) \to \ell_2^N(\Z^d),
\end{align}
where ${\rm Mult}_{\M U}: \M x[\cdot] \mapsto \M U \M x[\cdot]$ is the pointwise matrix-multiplication operator. 
This succession of operations yields the 
vector-valued impulse response $\M h[\cdot]=\frac{1}{\sqrt{N}}\sum_{n=1}^N \M u_n\delta[\cdot-\bk_n]$.
This filter is Parseval by construction because it is the composition of two Parseval operators.

As variant, we may also consider a reduced patch $\mathbb{K}_{N_0}=\{\bk_1,\dots,\bk_{N_0}\}$ with $N_0< N$,
which then results in a shorter Parseval filter $\M h[\cdot]=\frac{1}{\sqrt{N_0}}\sum_{n=1}^{N_0} \M u_n\delta[\cdot-\bk_n]$. The latter is parameterized by the ``truncated'' matrix $\M U_0=\begin{bmatrix}\M u_1& \cdots& \M u_{N_0}\end{bmatrix} \in \R^{N\times N_0}$, which is such that $\M U_0^\Tr \M U_0=\M I_{N_0}$ ($1$-tight frame property).

\subsection{Parametric $N$-to-$pN$ Parseval Module}
The concept here is essentially the same as in Section \ref{Sec:Parselva1toN},
except that we now have to use a larger ortho-matrix $\M U \in \R^{pN \times pN}$ and a patch neighorhood 
$\mathbb{K}_p=\{\bk_1,\dots,\bk_p\}$.
This then yields the multi-filter 
\begin{align}
\label{Eq:vectorized}
\Op T_{\M H} = {\rm Mult}_{\M U} \circ {\rm Patch}_{\mathbb{K}_p}: \ell^N_2(\Z^d) \to \ell_2^{ p N}(\Z^d),
\end{align}
which is guaranteed to have the Parseval property, based on the same arguments as before.
Its adjoint  is
\begin{align}
\Op T^\ast_{\M H}=\Op T^+_{\M H}={\rm Patch}^\ast_{\mathbb{K}_p} \circ {\rm Mult}_{\M U^\Tr}: \ell_2^{ p N}(\Z^d) \to \ell_2^{N}(\Z^d).
\end{align}
To identify the impulse response 
of ${\rm Mult}_{\M U} \circ {\rm Patch}_{\mathbb{K}_p}$, we partition $\M U=\begin{bmatrix} \M U_1 & \cdots \M U_p\end{bmatrix}$ into 
$p$ submatrices $\M U_i\in \R^{pN \times N}$, each associated with its shift $\V k_i$,
which yields
\begin{align}
\label{Eq:NtoNpimpulse}
\M H[\cdot]=\frac{1}{\sqrt{p}}\sum_{i=1}^p \M U_i \delta[\cdot-\bk_i].
\end{align}
By the orthonormality of the column vectors of $\M U$, 
we then explicitly evaluate
\begin{align}
\label{Eq:Autorr0}
(\M H^{\Tr \vee} \ast \M H)[\cdot]&=\frac{1}{p}\sum_{m=1}^p \sum_{n=1}^p \M U_m^\Tr \M U_n \delta[\cdot+\bk_m-\bk_n] \nonumber\\
&=\frac{1}{p}\sum_{n=1}^p  \M U_n^\Tr \M U_n \delta[\cdot]=\frac{1}{p}\sum_{n=1}^p   \M I_N \delta[\cdot]= \M I_N\delta[\cdot],
\end{align}
which confirms that $\Op T^\ast_{\M H}=\Op T^+_{\M H}$.

\subsection{Generalized Shift Composed with a Tight Frame}
With the view of extending \eqref{Eq:1toNelement} to vector-valued signals, we introduce the generalized shift (or scrambling) operator $\Op S^{\V K}: \ell_2^N(\Z^d) \to \ell^N_2(\Z^d)$, with translation parameter $\V K=(\bk_1, \dots,\bk_N) \in \Z^{d \times N}$, as
\begin{align}
\Op S^{\V K}\{\M x[\cdot]\}=
\begin{pmatrix}
\Op S^{\bk_1}\{x_1[\cdot]\} \\
\vdots\\
\Op S^{\bk_N}\{x_N[\cdot]\}
\end{pmatrix} 
= \begin{pmatrix}x_1[\cdot-\bk_1] \\
\vdots\\
x_N[\cdot-\bk_N]
\end{pmatrix}.\end{align}
It is the multichannel extension of the scalar shift by $\bk_0 \in \Z^d$ denoted by
$\Op S^{\bk_0}: x[\cdot] \mapsto x[\cdot- \bk_0]$.
The generalized shift is obviously LSI-Parseval and has the convenient semigroup property $\Op S^{\V K_0} \Op S^{\V K}=\Op S^{(\V K_0+\V K)}: \ell_2^N(\Z^d) \to \ell_2^N(\Z^d)$
with $\Op S^{\V K- \V K}=\Op S^{\M 0_N}=\Identity$, for any $\V K,\V K_0 \in \Z^{d \times N}$. This is parallel to the scalar setting where we have that $\Op S^{\V 0}=\Identity$ and
$\Op S^{\V k_0} \Op S^{\V k}=\Op S^{(\V k_0+\V k)} : \ell_2(\Z^d) \to \ell_2(\Z^d)$ for any $\V k,\V k_0 \in \Z^d$, so that $(\Op S^{\V k_0})^{-1}=\Op S^{-\bk_0}$
 with all shift operators being unitary.
 
Now, let $\M A=\begin{bmatrix}\M a_1 & \cdots & \M a_N \end{bmatrix}=\begin{bmatrix}\M b_1 & \cdots & \M b_M \end{bmatrix}^\Tr \in \R^{M \times N}$ with $M\ge N$ be a rectangular matrix such that
$ \M A^\Tr \M A=\M I_N$ (tight-frame property). 
The geometry is such that the column vectors  $\{\M a_n\}_{n=1}^N$ form an orthonormal family in
$\R^M$ (but not a basis unless $M=N$), while the row vectors $\{\M b_m\}_{m=1}^M$ form a $1$-tight frame of $\R^N$ that is the redundant 
counterpart of an ortho-basis. Here too, the defining property is energy conservation: $\sum_{m=1}|\langle\M b_m,\M x\rangle|^2=\langle \M A \M x,\M A \M x\rangle=\langle \M A^\Tr \M A \M x, \M x\rangle=\|\M x\|^2_2$ for all $\M x \in \R^N$ (Parseval), albeit in the simpler finite-dimensional setting.
 
Given such a tight-frame matrix $\M A\in \R^{M \times N}$ and a set of shift indices $\V K=(\bk_1, \dots,\bk_N) \in \Z^{d \times N}$, we then specify the operator
 \begin{align}
 \Op T_{\M H} = {\rm Mult}_{\M A} \circ \Op S^{\V K}=\M A\Op S^{\V K}: \ell^N_2(\Z^d) \to \ell_2^{ M}(\Z^d).
\end{align}
The matrix-valued frequency response of this filter is $\widehat{\M H}(\bw)=\M A\, {\rm diag}(\ee^{-\jj \langle\bw,\bk_1\rangle},$ $\dots,\ee^{-\jj \langle\bw,\bk_N\rangle})$, which is paraunitary, irrespectively of the choice of the shifts $\bk_m$.  

\subsection{$N$-to-$N$ Parseval Filters}
\label{Sec:ParesevalNtoN}
We know from Proposition \ref{Theo:PasevalFilters} that $\Op T_{\M H}: \ell^N_2(\Z^d) \to \ell^N_2(\Z^d)$ is a Parseval multi-filter
if and only if $\widehat{\M H}^\HTop(\bw)\widehat{\M H}(\bw)= \M I_N$ for all $\bw \in \mathbb{T}^d$.
By taking inspiration from the singular-value decomposition, this suggests the consideration of paraunitary elements of the form: 
$
\M U \widehat{\V \Lambda}(\bw)$,
$
\widehat{\V \Lambda}(\bw)\M V^\HTop$, or
$
\M U \widehat{\V \Lambda}(\bw)\M V^\HTop$,
where $\M U=\begin{bmatrix}\M u_1 &\cdots& \M u_N\end{bmatrix} \in \C^{N \times N}$ 
and $\M V=\begin{pmatrix}\M v_1 &\dots& \M v_N\end{pmatrix} \in \C^{N \times N}$ are unitary matrices,  and $\widehat{\V \Lambda}(\bw)={\rm diag}\big(\hat \lambda_1(\bw),\dots,\hat \lambda_N(\bw)\big)$ where the $\hat \lambda_n(\bw)$ are all-pass filters with $|\hat \lambda_n(\bw)|=1$ for all $\bw \in \mathbb{T}^d$.

Since our focus is on FIR filters, we invoke Proposition \ref{Prop:Shifts} to deduce that the only acceptable form of diagonal matrix
is $\widehat{\V \Lambda}(\bw)= {\rm diag}(\ee^{-\jj \langle\bw,\bk_1\rangle}, \dots,\ee^{-\jj \langle\bw,\bk_N\rangle})$ with shift parameter $(\bk_1, \dots,\bk_N)=\M K \in \Z^{d \times N}$, which is precisely the frequency response of the generalized shift operator $\Op S^{\V K}$. The resulting parametric Parseval operators are
$\M U\Op S^{\V K}$, $\Op S^{\V K}\M V^\HTop$,
and $\M U\Op S^{\V K}\M V^\HTop$. 
The impulse responses of these filters are sums of rank-1 elements with their support being specified by the $\bk_m$, which need not be distinct. Specifically, we have that
\begin{align}
\label{Eq:NtoNGeneral}
\M U\Op S^{\V K}\M V^\HTop=\Op T_{\M H}: \ell_2^N(\Z^d) \to \ell_2^N(\Z^d) \quad \mbox{ with } \quad \M H[\cdot]=\sum_{n=1}^N \M u_n \M v_n^\HTop \delta[\cdot- \bk_n],
\end{align}
which encompasses the two lighter filter variants by taking $\M U$ or $\M V$ equal to $\M I_N=\begin{bmatrix} \M e_1 & \cdots & \M e_N\end{bmatrix}$. Another canonical configuration of \eqref{Eq:NtoNGeneral} is obtained by taking $\M U=\M V$ which, as we shall see, makes an interesting connection with two classic factorization of paraunitary systems. While the operator $\M U\Op S^{\V K}\M V^\HTop$ is obviously a generalization of $\M U\Op S^{\V K}$, there is computational merit with the lighter version, especially in the context of composition.

\begin{proposition}
\label{Prop:Composition}
Let $\M W_1,\dots,\M W_{M+1} \in \C^{N \times N}$ and $\M U_1,\dots,\M U_{M+1} \in \C^{N \times N}$ be two series of orthogonal matrices,
and $\V K_1, \dots, \V K_M \in \Z^{d \times N}$ some corresponding shift indices.
Then, the composed parametric operators
\begin{align}
\label{Eq:Factor1}
\M W_{M+1} \Op S^{\V K_M}\M W_{M} \cdots \M W_3\Op S^{\V K_2}\M W_2\Op S^{\V K_1}\M W_1
\end{align}
and 
\begin{align}
\label{Eq:Factor2}
\M U^\HTop_{M+1}(\M U_M\Op S^{\V K_M}\M U^\HTop_M) \cdots (\M U_2\Op S^{\V K_2}\M U^\HTop_2) (\M U_1\Op S^{\V K_1}\M U^\HTop_1)
\end{align}
span the same family of $N$-to-$N$ Parseval multi-filters.
\end{proposition}
Indeed, by setting $\M W_1=\M U^\HTop_1$, $\M W_2=\M U^\HTop_2\M U_1$,
\dots, $\M W_{M}=\M U^\HTop_M\M U_{M-1}$
and $\M W_{M+1}=\M U^\HTop_{M+1}\M U_M$, we can use \eqref{Eq:Factor1} to replicate 
\eqref{Eq:Factor2}. Again, the key is that the multiplication (composition) of two unitary matrices yields another unitary matrix.
Conversely, \eqref{Eq:Factor2} reproduces \eqref{Eq:Factor1} if we set
$\M U^\HTop_1=\M W_1$, $\M U^\HTop_2=\M W_2\M U^\HTop_1= \M W_2 \M W_1$,\dots, $\M U^\HTop_M=\M W_{M-1} \M U^\HTop_{N-1}=\M W_{M-1}\M W_{M-2} \cdots \M W_1$, and $\M U^\HTop_{M+1}=\M W_M\M U_N^\HTop=\M W_M \cdots \M W_1$.

\subsection{Projection-based Parseval Filterbanks}
\label{Sec:Projection}
These $N$-to-$N$ filterbanks are parameterized by a projection matrix $\M P$. They are multi-dimensional adaptations of classic canonical structures that were introduced by Vaidyanathan and others for the factorization of paraunitry matrices for $d=1$ \cite{Soman1993}.
To explain the concept, we recall that a matrix $\M P$ is a member of $\mathbb{P}(N,k)$ (the set of all orthogonal projection matrices in $\R^N$ of rank $k$) if and only if it fulfils the following conditions:
\begin{enumerate}
\item Rank: $\M P \in \R^{N \times N}$ with ${\rm rank}(\M P)=k$.
\item Idempotence: $\M P\M P=\M P$.
\item Symmetry: $\M P^\Tr =\M P$, which together with Item 2, implies that $\M P$ is an ortho-projector.
\end{enumerate}
Any $\M P \in  \mathbb{P}(N,k)$ can be parameterized 
as $\M P=\sum_{n=1}^k \M u_n \M u_n^\Tr={\rm Proj}_{{\rm span}(\M u_1,\dots,\M u_k)}$,
where $\M u_1,\dots,\M u_k$ is a set of orthogonal vectors in $\R^N$. Since $\M P$ is an ortho-projector,
it induces the direct-sum decomposition $\R^N={\rm Ran}(\M P) \oplus {\rm Ker}(\M P)$, where the members
of ${\rm Ran}(\M P)$ are eigenvectors of $\M P$ with eigenvalue $1$ (projection property), while the members of ${\rm Ker}(\M P)={\rm Ran}(\M P)^\perp$ are eigenvectors with eigenvalue $0$. The parametrization of $\M P$ then simply follows from the SVD, with the vectors
$\M u_1,\dots,\M u_k$ being any set of orthogonal members of ${\rm Ran}(\M P)$.
In particular, the rank-1 ortho-projectors
are parameterized by a single unit vector, with $\mathbb{P}(N,1)=\{\M u\M u^\Tr \in \R^{N \times N} \mbox{s.t. } \|\M u\|_2=1\}$.
Finally, two projection matrices $\M P\in\mathbb{P}(N,k)$ and
$\widetilde{\M P}\in\mathbb{P}(N,N-k)$ are said to be complementary if  $\M P+\widetilde{\M P}=\M I_N$. In fact, $\M P_k\in\mathbb{P}(N,k)$ has a single complementary projector that is given by $\M I_N-\M P \in \mathbb{P}(N,N-k)$.

A basic FIR-Parseval projection element is characterized by a matrix impulse response of the form
\begin{align}
\label{Eq:ProjectionFilter}
\M H_{\M P,\V k_1}[\cdot]=(\M I_N -\M P)\delta[\cdot]+\M P \delta[\cdot-\V k_1],
\end{align}
where $\M P \in \R^{N \times N}$ is a projection matrix and $\V k_1\in[-1,1]^d\backslash\{\V 0\}$ is some elementary (multidimensional) unit shift. 
We shall refer to such a structure by PROJ-$k$, with $k$ being the rank the projector. 
The impulse response of a PROJ-$1$ element is
\begin{align}
\label{Eq:HouseholderFilter}
\M H_{\M u,\V k_1}[\cdot]=(\M I_N - \M u\M u^\Tr)\delta[\cdot] + \M u\M u^\Tr \delta[\cdot-\V k_1]
\end{align}
which, similarly to a Householder matrix, can be parameterized by a single unit vector $\M u$.
More generally for PROJ-$k$, the condition $\M P\in\mathbb{P}(N,k)$ translates into the existence of an ortho-matrix
$\M U=[\M u_1 \cdots \M u_N]\in \R^{N \times N}$ such that $\M P=\sum_{n=1}^k \M u_n \M u_n^{\Tr}$
and $(\M I_N - \M P)=\sum_{n=k+1}^N \M u_n \M u_n^{\Tr}$. This then allows us
to express the convolution operator $\M x[\cdot]  \mapsto (\M H_{\M P,\V k_1}\ast \M x)[\cdot]$ as
$\Op T_{\M H_{\M P,\V k_1}}=\M U\Op S^{\V K_1}\M U^\HTop$ with a scrambling matrix $\V K_1$ that consists of two shifts only: $\V k_1$ repeated $k$ times, and  $\V 0$ (zero) for the remaining entries.
Consequently, \eqref{Eq:ProjectionFilter} and \eqref{Eq:HouseholderFilter} are two special cases of \eqref{Eq:NtoNGeneral}, which confirms their Parseval  property.

While the $N$-to-$N$ scheme we presented in Section \ref{Sec:ParesevalNtoN} is more general and also suggests some natural computational streamlining (see Proposition \ref{Prop:Composition}), arguments can be made in favor of the use of PROJ-$k$ filtering components, each of which has a minimal support of size $2$.

The strongest argument is theoretical but only holds for $d=1$  \cite{Gao2001}. Specifically, it has been shown that
all Parseval filters of a fixed McMillan degree (i.e., the total number of delays required to implement the filterbank) admit a factorization in terms of Proj-$1$ elements \cite{Soman1993}. Likewise, any filterbank with filters of fixed support $M$ admits a factorization in terms of Proj-$k$ elements, which ensures that such a parametrization is complete \cite{Turcajova1994b}. The tricky part in this latter type of factorization is that it also requires the adjustment of the rank $k_i$ of each component. 

Unfortunately, such results do not generalize to higher dimensions because of the lack of a general polynomial factorization theorem for $d>1$. Simply stated, this means that there are many multidimensional filters that cannot be realized from a composition of elementary filters of size $2$.
For $d=1$, the elementary shift in \eqref{Eq:ProjectionFilter} and \eqref{Eq:HouseholderFilter} is set to $k_1=1$, but it is not clear how to proceed systematically in higher dimensions.

In the context of a convolutional neural network where many design choices are ad hoc, the lack of guarantee of completeness among all Parseval filters of size $M$ (one arbitrary family of filters among many others) is not particularly troublesome. The more important issue is to be able to exploit the available degrees of freedom by adjusting the parameters for best performance during the training procedure.
This is achieved effectively for $d=2$ in the block-convolution orthogonal parametrization (BCOP) framework \cite{Li2019}, which relies on the composition of PROJ-$k$ with $\bk_1\in \{(0,1),(1,0)\}$ (in alternation). 
By formulating the training problem with twice the number of channels (half of which are dummy and constrained to have zero output) 
with $\M P\in\mathbb{P}(2N,N)$, the authors are also able to optimally adjust the parameter $k$ (rank of the projector) for each unit.

\section{Application to Denoising and Image Reconstruction}
\label{sec:PnPReconstruction}

We now discuss  the application of Parseval filterbanks to biomedical image reconstruction. 
Specifically, we shall rely on 1-Lipschitz neural networks that use Parseval convolution layers and that are trained for the denoising of a representative set of images.

Depending on the context, the image to be reconstructed is described as a signal $s[\cdot] \in \ell_2(\Z^d)$ or as the vector 
$\M s=(s[\bk])_{\bk \in \Omega}\in \R^{K}$, where $\Omega \subset \Z^d$ is a region of interest composed of a finite number $K$ of pixels.
Our computational task is to recover $\M s \in \mathbb{R}^{K}$ from the noisy measurement vector
\begin{align}
\label{first_equ}
 \mathbf{y} = \mathbf{As} +{\mathbf{n}} \in \mathbb{R}^{M},
\end{align}
where $\mathbf{n}$ is some (unknown) noise component and where $\mathbf{A} \in \mathbb{R}^{M \times K}$ is the system matrix that models the physics of the acquisition process. A simplified version of \eqref{first_equ} with $M=K$ and $\M A=\M I$ (identity) is the basic denoising problem, where the task is to recover $\M s$ from the noisy signal
\begin{align}
\label{first_equ0}
 \mathbf{z} = \mathbf{s} + \mathbf{n} \in \mathbb{R}^{K}.
\end{align}
\subsection{From Variational to Iterative Plug-and-Play Reconstruction}
To make signal-recovery problems well-posed mathematically, one usually incorporates prior knowledge about the unknown image $\M s$ by imposing regularity constraints on the solution.
This leads to the variational reconstruction 
\begin{equation}
\label{eq:obj_function}
    \M s^\ast =\argmin_{\M s \in \mathbb{R}^{K}} \left(J(\mathbf{y}, \mathbf{A}\M s) + R(\M s)\right),
\end{equation}
where $J\colon \R^{M} \times \R^{M} \to \R^{+}$ is a data-fidelity term and $R\colon \R^{K} \to \R^{+}$ is a regularization functional that penalizes ``non-regular'' solutions.
If $\M s \mapsto J(\mathbf{y}, \mathbf{A}\M s)$ is differentiable and $R$ is convex, then \eqref{eq:obj_function} can be solved by the iterative forward-backward splitting (FBS) algorithm \cite{Combettes2005} with
\begin{equation}
    \M s^{k+1}=\prox_{\alpha R}\bigl\{\M s^{k}-\alpha \pmb{\nabla} J(\mathbf{y}, \mathbf{As}^{k})\bigr\}.
\end{equation}
Here, $\pmb{\nabla} J(\mathbf{y}, \mathbf{As})$ is the gradient of $J$ with respect to $\M s$, $\alpha \in \R$ is the stepsize of the update, and $\prox_{R}$ is
the proximal operator of $R$
defined as \begin{align}
\prox_{R}\{\mathbf{z}\}= \argmin_{\M s \in \R^{K}} \left(\tfrac12 \|\M z-\mathbf{s}\|^{2} + R(\M s)\right).
\end{align}
An important observation is that $\prox_{R}: \R^K \to \R^K$ actually returns the solution of the denoising problem 
with a variational formulation that
is a particular case of \eqref{eq:obj_function} with $\M A=\M I$ and the quadratic data term $J(\mathbf{z}, \M s)= \tfrac12 \|\M z-\mathbf{s}\|^{2}$.

The philosophy of PnP algorithms \cite{Venkatakrishnan2013plug} is to replace $\prox_{\alpha R}$ with an off-the-shelf denoiser $\Op D\colon \mathbb{R}^{K} \to \mathbb{R}^{K}$.
While not necessarily corresponding to an explicit regularizer $R$, this approach has led to improved results in image reconstruction, as shown in \cite{Ryu2019plug, Sun2021, Ye2018}.
The convergence of the PnP-FBS iterations
\begin{equation}
\label{eq:pnp_fbs}
    \M s^{k+1}=\Op D\bigl\{\M s^{k}-\alpha \pmb{\nabla} J(\mathbf{y}, \mathbf{As}^{k})\bigr\}
\end{equation}
can be guaranteed \cite[Proposition 15]{Hertrich2021} if 
\begin{itemize}
    \item the denoiser $\Op D$ is averaged, which means that is takes the form $\Op D = \beta \Op R + (1 - \beta)\operatorname{Id}$ with $\beta \in(0, 1)$ and an operator $\Op R: \R^K \to \R^K$ such that ${\rm Lip}(\Op R)\le 1$; 
    \item the data term $J(\mathbf{y}, \mathbf{H}\cdot)$ is convex, differentiable with $L$-Lipschitz gradient, and $\alpha \in (0, 2/L)$.
\end{itemize}
Moreover, it is possible to prove that the solution(s) of the PnP algorithm 
satisfies the properties expected of a faithful reconstruction. The first such property is a joint form of consistency between the reconstructed image $\M s^\ast$ (outcome of the algorithm) and the measurement $\M y$ (input).

\begin{proposition}\label{prop:stability1}
Let $\M s_1^{*}$ and $\M s_2^{*}$ be fixed points of \eqref{eq:pnp_fbs} for measurements $\vec y_1$ and $\vec y_2$, respectively.
If the operator $\Op D$ is averaged with $\beta \leq 1/2$ and $J(\mathbf{y}, \mathbf{As}) = \frac{1}{2}\|\mathbf{y} - \mathbf{As}\|_2^{2}$, then it holds that
\begin{equation}\label{eq:EstForward}
\|\vec A\M s_1^{*} - \vec A\M s_2^{*}\| \leq \|\vec y_1 - \vec y_2\|.
\end{equation}
\end{proposition}
\begin{proof}
If $\Op D$ is $\beta$-averaged with $\beta \leq 1/2$, then $(2\Op D - \operatorname{Id})$ is 1-Lipschitz since\begin{align}
\|(2\Op D-\operatorname{Id})\{\mathbf{z}_1 - \mathbf{z}_2\}\| & = \|2\beta(\Op R\{\vec z_1\} - \Op R\{\vec z_2\} + (1-2\beta)(\vec z_1 - \vec z_2)\| \notag\\
& \leq 2\beta \|\Op R\{\vec z_1\} - \Op R\{\vec z_2\}\| + (1-2\beta)\|\vec z_1 - \vec z_2\| \notag \\
& \leq \|\vec z_1 -\vec z_2\|, \quad \forall \vec z_1, \vec z_2 \in \mathbb{R}^{K}.
\end{align}
Using this property, we get that
\begin{align}
\|(2\Op D-\operatorname{Id})\{\M s_1^{*} - \alpha \pmb{\nabla} J(\vec A \M s_1^{*}, \vec y_1)\} - (2\Op D-\operatorname{Id})\{\M s_2^{*} - \alpha \pmb{\nabla} J(\vec A \M s_2^{*}, \vec y_2)\}\| \notag \\
\leq \|(\M s_1^{*} - \alpha \pmb{\nabla}J(\vec A \M s_1^{*}, \vec y_1)) - (\M s_2^{*} - \alpha \pmb{\nabla} J(\vec A \M s_2^{*}, \vec y_2))\|.
\end{align}
From the fixed-point property of  $\M s_1^{*}$ and $\M s_2^{*}$, it follows that
\begin{align}
\|2(\M s_1^{*} - \M s_2^{*}) - (\M s_1^{*} - \alpha \pmb{\nabla} J(\vec A \M s_1^{*}, \vec y_1)) + (\M s_2^{*} - \alpha \pmb{\nabla} J(\vec A \M s_2^{*}, \vec y_2))\| \notag \\
\leq \|(\M x_1^{*} - \alpha \pmb{\nabla} J(\vec A \M s_1^{*}, \vec y_1)) - (\M s_2^{*} - \alpha \pmb{\nabla} J(\vec A \M s_2^{*}, \vec y_2))\|.
\end{align}
Next, we use the fact that $\pmb{\nabla} J(\vec A \M s, \vec y) = \vec A^\Tr(\vec A \M s - \vec y)$ and develop both sides as
\begin{align}
\langle \M s_1^{*} - \M s_2^{*}, \vec A^\Tr(\vec A \M s_2^{*} - \vec y_2) - \vec A^\Tr(\vec A \M s_1^{*} - \vec y_1)\rangle \geq 0.
\end{align}
Finally, we move $\vec A^\Tr$ to the other side of the inner product and invoke the Cauchy-Schwartz inequality to get that
\begin{align}
\|\vec A (\M s_1^{*} - \M s_2^{*})\| \|\vec y_1 - \vec y_2\| \geq \langle \vec A(\M s_1^{*} - \M s_2^{*}), \vec y_1 - \vec y_2 \rangle \geq \|\vec A(\M s_1^* - \M s_2^{*})\|^{2},
\end{align}
which is equivalent to \eqref{eq:EstForward}.
\end{proof}

When $\vec A$ is invertible, \eqref{eq:EstForward} yields the direct relation 
\begin{equation}\label{eq:H_invert}
\|\M s_1^{*} - \M s_2^{*}\| \leq \frac{1}{\sigma_{\text{min}}(\vec A^\Tr \vec A)}\|\vec y_1 - \vec y_2\|.
\end{equation}
It ensures that the iterative reconstruction algorithm is itself globally Lipschitz stable. In other words, a small deviation of the input 
can only result in a limited deviation of the output, which intrinsically provides protection against hallucinations.
Under slightly stronger constraints on $\Op D$, we have a comparable result for non-invertible $\M A$.
\begin{proposition}\label{prop:stability2}
    In the setting of Proposition \ref{prop:stability1} and for a $L_0$-Lipschitz denoiser $\Op D$ with $L_0<1$, it holds that
    \begin{equation}\label{eq:EstCont}
    \|\M s_1^{*} - \M s_2^{*}\| \leq        \frac{\alpha \|\vec A\|L_0}{1 - L_0}\|\vec     y_1 - \vec y_2\|.
    \end{equation}
\end{proposition}

\begin{proof}
\begin{align}
\|\M s_1^{k} - \M s_2^{k}\| & = \|\Op D\{\M s_1^{k-1} - \alpha \M A^\Tr(\vec A \M s_1^{k-1} - \vec y_1)\} - \Op D\{\M s_2^{k-1} - \alpha \vec A^\Tr(\vec A \M s_2^{k-1} - \vec y_2)\}\| \notag \\
& \leq L _0\|(\vec I - \alpha \vec A^\Tr \vec A)(\M s_1^{k-1} - \M s_2^{k-1}) - \alpha \vec A^\Tr(\vec y_1 - \vec y_2)\| \notag \\
& \leq L_0 \|\M s_1^{k-1} - \M s_2^{k-1}\| + \alpha L_0 \|\vec A\|\|\vec y_1 - \vec y_2\| \notag \\
& \leq L_0^{2}\|\M s_1^{k-2} - \M s_2^{k-2}\| + \alpha \|\vec A\|(L_0 + L_0^{2}) \|\vec y_1 - \vec y_2\| \notag \\
& \leq L_0^{k} \|\M s_1^{0} - \M s_2^{0}\| + \alpha \|\vec A\| \|\vec y_1 - \vec y_2\| \sum_{n=1}^{k} L_0^{n}.
\end{align}
Taking the limit $k \rightarrow \infty$, we get that $\| \M s_1^{*} - \M s_2^{*}\| \leq \frac{\alpha \|\vec A\|L_0}{1 - L_0} \|\vec y_1 - \vec y_2\|$.
\end{proof}

Since it is formulated as a data fitting problem, the reconstruction \eqref{eq:obj_function} generally has better data consistency than the one provided by end-to-end neural-network frameworks that directly reconstruct $\M s$ from $\mathbf{y}$ \cite{Jin2017,Mccann2017Convolutional,Wang2020,Lin2021artificial}.
Those latter approaches are also known to suffer from stability issues \cite{Antun2020}. More importantly, they have been found to remove or hallucinate structure \cite{Nataraj2020, Muckley2021MRIChallenge}, which is unacceptable in diagnostic imaging.
The usage of empirical PnP methods without strict Lipschitz control within the loop is also subject to caution, as they do 
not offer any guarantee of stability.
By contrast, the PnP approach \eqref{eq:pnp_fbs} with averagedness constraints comes with the stability bounds \eqref{eq:EstForward}, \eqref{eq:H_invert} and \eqref{eq:EstCont}.
This is a step toward reliable deep-learning-based image reconstruction
as it intrinsically limits the ability of the method to overfit and to hallucinate. 

\subsection{Learning an Averaged Denoiser for PnP}
Our approach to improve upon classic image reconstruction is to 
learn the operator $\Op D=\beta \Op R + (1 - \beta)\operatorname{Id}$ in \eqref{eq:pnp_fbs}. We pretrain it for the best performance in the denoising scenario \eqref{first_equ0}.
To that end, we impose the structure of the 1-Lip LSI operator $\Op R$ as an $L$-layer convolutional
neural network with all intermediate layers being composed of the same number ($N$) of feature channels. Specifically, by reverting back to the notation of Section \ref{Sec:VecValued}, we have that
 $\Op  R: \ell_2(\Z^d) \to \ell_2(\Z^d)$ with
\begin{equation}
\Op R = \Op T_{\vec H_L} \circ \V \sigma_{L} \circ \Op T_{\vec H_{L-1}} \circ \V \sigma_{L-1}  \circ ...\circ \Op T_{\vec H_{2}} \circ \V \sigma_2 \circ \Op T_{\vec H_1},
\end{equation}
where $\Op T_{\vec H_k}$ are LSI operators with matrix-valued impulse response $\M H_k[\cdot]$ and $\V \sigma_k=(\sigma_{k,1},\dots,\sigma_{k,N})$ are pointwise 
nonlinearities with the shared activation profile $\sigma_{k,n}: \R \to \R$ within each feature channel. As for the domain and range of the operators, we have that $\Op T_{\vec H_1}: \ell_2(\Z^d) \to \ell^N_2(\Z^d)$ and $T_{\vec H_L}: \ell^N_2(\Z^d)\to \ell_2(\Z^d)$ for the input and output layers,
while $T_{\vec H_k}: \ell^N_2(\Z^d) \to \ell^N_2(\Z^d)$ for $k=2,\dots,(L-1)$. Likewise, $\V \sigma_k:\ell^N_2(\Z^d) \to \ell^N_2(\Z^d)$, with the effect of the nonlinear layer being described by
\begin{align}
\forall \bk \in \Z^d: \quad 
\V \sigma_k\left\{\begin{pmatrix}x_1[\cdot] \\
\vdots\\
x_N[\cdot]
\end{pmatrix}\right\}[\bk]=\begin{pmatrix}\sigma_{k,1}(x_1[\bk]) \\
\vdots\\
\sigma_{k,N}(x_N[\bk])
\end{pmatrix}
\end{align}
with activation functions $\sigma_{k,n}: \R \to \R$  for $n=1,\dots,N$.

Since the Lipschitz constant of the composition of two operators is bounded by the product of their individual Lipschitz constant, we have that
\begin{equation}
    \text{Lip}(\Op R) \leq \text{Lip}(\Op T_{\vec H_1}) \prod_{k=2}^{L} \text{Lip}(\V \sigma_k)\text{Lip}(\Op T_{\vec H_k}),
\end{equation}
which means that we can ensure that $\operatorname{Lip}(\Op R) \leq 1$ by constraining each $\V \sigma_k$ and $\Op T_{\vec H_k}$ to be 1-Lipschitz. 

\subsubsection{Specification of 1-Lip Convolution Layers}
We consider two ways of enforcing $\text{Lip}(\Op T_{\vec H_k})= 1$. Both are supported by our theory.

\begin{itemize}
    \item Spectral normalization (SN) \cite{Ryu2019plug}: During the learning process, we repeatedly renormalize the denoising filters $\vec H_k$ by dividing them by their spectral norm ${\rm Lip}(\Op T_{\vec H_k})=\|\Op T_{\vec H_k}\|=\sigma_{\sup,\M H_k}$ (see Theorem \ref {Theo:L2VecBound}).
    \item BCOP \cite{Li2019}: The Parseval filters $\Op T_{\vec H_k}$ are parameterized explicitly using orthogonal matrices $\M U_k\in \R^{N \times N}$,
    as described in Sections \ref{Sec:ParesevalNtoN}-\ref{Sec:Projection}. We use the implementation provided by the BCOP framework of Li et al.
    As for the last $N$-to-$1$ multifilter $\Op T_{\vec H_L}$, it is not literally Parseval, but rather the adjoint of a Parseval operator, which preserves the $1$-Lip property as well.
\end{itemize}

\subsubsection{Specification of 1-Lip Activation Functions}
The Lipschitz constant of a nonlinear scalar activation $f: \R \to \R$ is given by
\begin{align}
{\rm Lip}(f)=\sup_{t \in \R} |\frac{\dint f(t)|}{\dint t}|.
\end{align}
This result can then be applied to the full nonlinear layer $\V \sigma_k:\ell^N_2(\Z^d) \to \ell^N_2(\Z^d)$ through the pooling formula \eqref{Eq:Lipnonlinearlayer}.
\begin{proposition}
Let $\V f: \ell_2^N(\Z^d)\to\ell_2^N(\Z^d)$ be a generic pointwise nonlinear mapping specified by
\begin{align}
\V f\big\{ \M x[\cdot]\big\}[\V k]=\begin{pmatrix} f_{\bk,1}(x_1[\bk]) \\
\vdots\\
f_{\bk,N}(x_N[\bk]) 
\end{pmatrix}, \quad \bk \in \Z^d
\end{align}
Then, $\V f: \ell_2^N(\Z^d)\to\ell_2^N(\Z^d)$ is Lipschitz continuous if and only if all the component-wise transformations $f_{\bk,n}:\R \to \R$, with $(\bk,n) \in \Z^d \times \{1,\dots,N\}$ are Lipschitz-continuous. Its Lipschitz constant is then given by 
\begin{align}
\label{Eq:Lipnonlinearlayer}
{\rm Lip}(\V f)=L_\V f=\sup_{(\bk,n) \in \Z^d \times \{1,\dots,N\}} {\rm Lip}(f_{\bk,n})<\infty.
\end{align}\end{proposition}
\begin{proof}
Under the assumption that the $f_{\bk,n}$ are Lipschitz continuous, for any $\M x,\M y \in \ell_2^N(\Z^d)$ we have that
\begin{align*}
\|\V f\{\M y\}-\V f\{\M x\}\|^2_{\ell_2^N}&=\sum_{n=1}^N\sum_{\bk \inZ^d} \big|f_{\bk,n}(y_{n}[\bk]) - f_{\bk,n}(x_{n}[\bk])
\big|^2\\
&\le \sum_{n=1}^N\sum_{\bk \inZ^d} {\rm Lip}(f_{\bk,n})^2 \big|y_{n}[\bk] - x_{n}[\bk]\big|^2\\
&\le \left(\sup_{(\bk,n) \in \Z^d \times \{1,\dots,N\}} {\rm Lip}(f_{\bk,n})\right)^2 \  \|\M y- \M x\|^2_{\ell_2^N},
\end{align*}
which proves that ${\rm Lip}(\V f)\le L_\V f$. From the definition of the supremum, for any $\epsilon>0$, there exists some
$(\bk_0,n_0) \in \Z^d \times \{1,\dots,N\}$ such that ${\rm Lip}(\V f)\le L_\V f\le (1 + \epsilon) L_0$ with $L_0={\rm Lip}(f_{\bk_0,n_0})$.
Likewise, since $f=f_{\bk_0,n_0}: \R \to \R$ is $L_0$-Lipschitz continuous, for any $\epsilon'>0$ there exist some $x,y \in \R$ with $x\ne y$ such that $(1+\epsilon')|f(y)-f(x)|\ge L_0 |y-x|$. We then consider the corresponding (worst-case) signals
$\tilde{\M x}=x \M e_{n_0}\delta[\cdot- \bk_0]$ and $\tilde{\M y}=y \M e_{n_0}\delta[\cdot- \bk_0]$, for which
we have that ${\rm Lip}(\V f) \|\tilde{\M y}-\tilde{\M x}\|_{\ell_2^N}\ge \|\V f(\tilde{\M y})-\V f(\tilde{\M x})\|_{\ell_2^N}\ge \frac{L_{\V f}}{(1+\epsilon)(1+\epsilon')} \|\tilde{\M y}-\tilde{\M x}\|_{\ell_2^N}$. Since $\epsilon'$ and $\epsilon$ can be chosen arbitrarily small, the Lipschitz bound is sharp with ${\rm Lip}(\V f)=L_\V f$.
The same kind of worst-case signals can also be used to show the necessity of the Lipschitz continuity of each $f_{\bk,n}: \R \to \R$.
\end{proof}

Accordingly,  in our experiments, we have considered two configurations.
\begin{itemize}
    \item Fixed activation as a rectified linear unit (ReLU) with ${\rm Lip}({\rm ReLU})=\|\Indic_+\|_{L_\infty}=1$.
    \item Learnable linear spline (LLS) \cite{Ducotterd2024}, with learned activations $\sigma_{k,n}: \R \to \R$ s.t.\ ${\rm Lip}(\sigma_{k,n})=1$. These nonlinearities  are shared within each convolution channel $(k,n)\in \{2,\dots,L\} \times \{1,\dots,N\}$. They are parameterized using linear B-splines subject to a second-order total-variation regularization that promotes continuous piecewise-linear solutions with the fewest linear segments \cite{Bohra2020b,Unser2019c}.
\end{itemize}

\subsubsection{Image Denoising Experiments}
We train 1-Lip denoisers with $L=8$, $N=64$, and filters of size $(3 \times 3)$. The training dataset consists of 238400 patches of size $(40 \times 40)$ taken from the BSD500 image dataset \cite{Arbelaez2011}.
All noise-free images $\M s$ in \eqref{first_equ0} are normalized to take values in $[0, 1]$. They are then corrupted with additive Gaussian noise of standard deviation $\sigma$ to train the denoiser $\Op D$ for the regression task $\Op D\{\M z\}\approx \M s$.
The performance on the BSD68 test set is provided in Table~\ref{table:denoising_perf} for $\sigma =  5/255, 10/255$.
The general trend for each experimental condition is the same: The Parseval filters parameterized by BCOP consistently outperform the $1$-Lip filters obtained by simple spectral normalization. There is also a systematic benefit  in the utilization of learned $1$-Lip spline activations (LLS), as compared to the standard ReLU design.

\begin{table}[tb]
\centering
\caption{PSNR and SSIM on BSD68 for two noise levels.}
\label{table:denoising_perf}
\begin{tabular}{lcccc}
\hline
\hline
Noise level & \multicolumn{2}{c}{ $\sigma = 5/255$ } & \multicolumn{2}{c}{ $\sigma = 10/255$ } \\
Metric & PSNR & SSIM & PSNR & SSIM \\
\hline
ReLU-SN & 35.78 & 0.9297 & 31.48 & 0.8533\\
ReLU-BCOP & 36.10 & 0.9386 & 31.92 & 0.8735 \\
LLS-SN & 36.68 & 0.9504 & 32.36 & 0.8883\\
LLS-BCOP & \textbf{36.86} & \textbf{0.9546} & \textbf{32.55} & \textbf{0.8962} \\
\hline\hline
\end{tabular}
\vspace{.2cm}
\end{table}

\subsection{Numerical Results for PnP-FBS}\label{sec:exp_pnp}
We now demonstrate the deployment of our learned denoisers in the PnP-FBS algorithm for image reconstruction. To that end, we select the data-fidelity term as $J(\mathbf{y}, \mathbf{As}) = \frac{1}{2}\|\mathbf{y} - \mathbf{As}\|_2^{2}$, where the matrix $\M A$ simulates the physics of biomedical image acquisitions \cite{McCann2019}.
To ensure the convergence of \eqref{eq:pnp_fbs}, we set $\alpha = 1/ \|\vec A^\Tr \vec A\|$. Our denoiser is defined as $\Op D = \beta \Op R + (1 - \beta)\text{Id}$, while the constant $\beta \in [0, 1)$ and the training noise level $\sigma \in \{5/255, 10/255\}$ are tuned for best performance. 
In our experiments, we noticed that the best $\beta$ is always lower than 1/2, which means that the mathematical assumptions for Proposition \ref{prop:stability1} are met. We also compare our reconstruction algorithms with the classic total-variation (TV) method \cite{Chambolle2004}.

In our MRI experiment, the goal is recover $\M s$ from $\vec y=\mathbf{M F s}+\mathbf{n} \in \mathbb{C}^{M}$, where $\mathbf{M}$ is a subsampling mask (identity matrix with some missing entries), $\mathbf{F}$ is the discrete Fourier-transform matrix, and $\mathbf{n}$ is a realization of a complex-valued Gaussian noise characterized by $\sigma_\mathbf{n}$ for the real and imaginary parts. 
We investigated three $k$-space sampling schemes (random, radial, and Cartesian=uniform along the horizontal direction), each giving rise to a specific sub-sampling mask.

The reconstruction performance for various $k$-space sampling configurations and design choices is reported in Table~\ref{table:reconstruction_performance}. Similarly to the denoising experiment, BCOP always outperforms SN, while LLS brings additional improvements. Our CNN-based methods generally perform better than TV (standard reconstruction algorithm), while they essentially offer the same theoretical guarantees (consistency and stability) \cite{delAguilaPla2023}. The only notable exception is the TV-regularized reconstruction of Brain with Cartesian sampling, which is of better quality than the one obtained with SN-ReLU.
The results for Brain and Bust with the Cartesian mask are shown in Figures~\ref{fig:mri} and \ref{fig:mri2}, respectively. 
In the lower panel of Figure~\ref{fig:mri}, we observe stripe-like structures in the zero-fill reconstruction.
These are typical aliasing artifacts that result from the subsampling in the horizontal direction in Fourier space.
They are significantly reduced with the help of TV (which is routinely used for that purpose) as well as in the LLS-BCOP reconstruction, which overall yields the best visual quality.

\begin{figure}[tpb]
    \centering
    \includegraphics[scale=0.40]{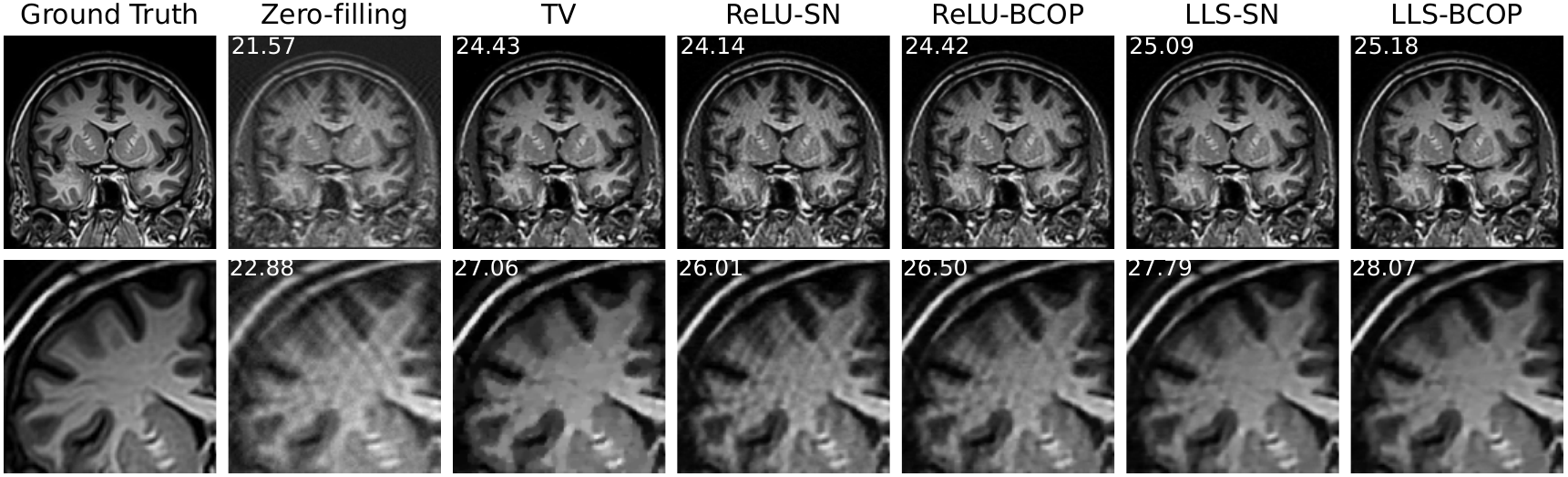}
    \caption{Ground truth, zero-fill reconstruction $\vec H^{\Tr} \vec y$, and PnP-FBS reconstruction using several network parameterizations on the Brain image with the Cartesian mask.\\ Lower panel: zoom of a region of interest. The SNR is evaluated with respect to the groundtruth (left image)
and is overlaid in white.
}
    \label{fig:mri}
    \vspace{.35cm}
\end{figure}

\begin{figure}[tpb]
    \centering
    \includegraphics[scale=0.40]{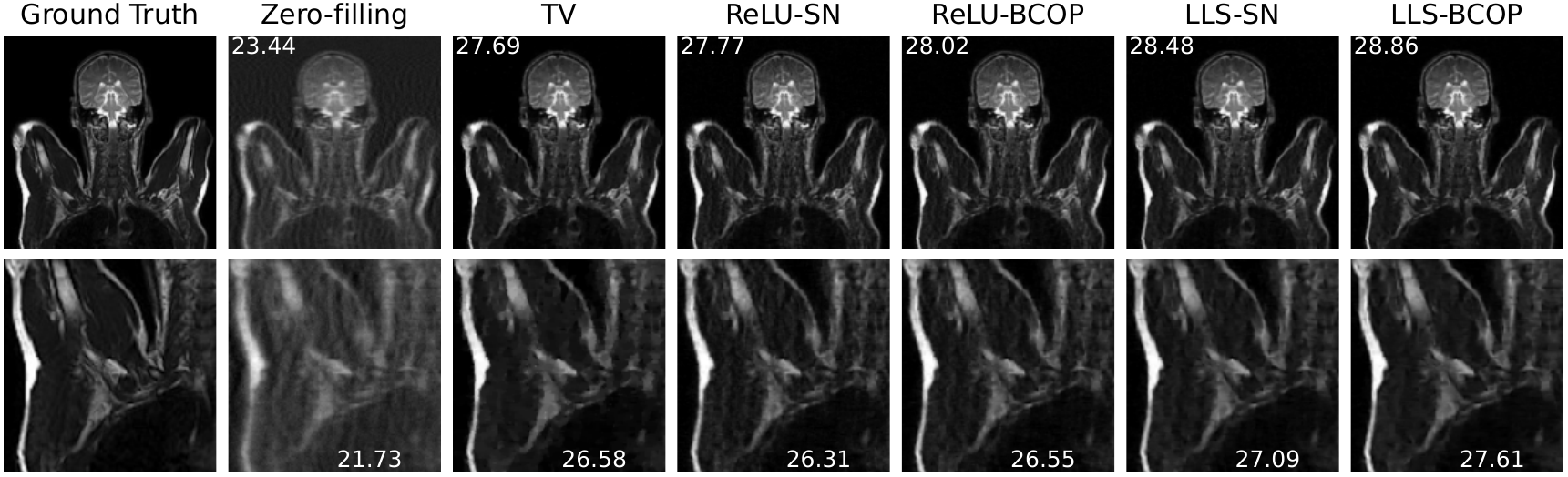}
    \caption{Ground truth, zero-fill reconstruction $\vec H^{\Tr} \vec y$, and PnP-FBS reconstruction using several network parameterizations on the Bust image with the Cartesian mask.\\ Lower panel: zoom of a region of interest. The SNR is evaluated with respect to the groundtruth (left image)
and is overlaid in white.
}
    \label{fig:mri2}
    \vspace{.35cm}
\end{figure}

\begin{table}[tp]
\centering
\caption{PSNR and SSIM for the MRI reconstruction experiment.}%
\label{table:reconstruction_performance}
\begin{tabular}{lcccccc}
  \hline
   \hline

    Subsampling mask  & \multicolumn{2}{c}{Random} & \multicolumn{2}{c}{Radial} & \multicolumn{2}{c}{Cartesian}
    \\
       Image type & Brain &  Bust  & Brain & Bust   & Brain & Bust
    \\ \hline
    Zero-filling & 24.68 & 27.31 & 23.85 & 25.13 & 21.57 & 23.44\\
    TV & 30.37 & 32.29 & 29.46 & 31.58 & 24.43 & 27.69 \\
    ReLU-SN & 32.45 & 33.36 & 30.92 & 32.33 & 24.14 & 27.77\\
    ReLU-BCOP & 32.53 & 33.67 & 30.93 & 32.72 & 24.42 & 28.02 \\
    LLS-SN & 33.34 & 34.32 & 31.82 & 33.35 & 25.09 & 28.48\\
    LLS-BCOP & \textbf{33.61} & \textbf{34.67} & \textbf{32.09} & \textbf{33.72} & \textbf{25.18} & \textbf{28.86} \\\hline
    \hline
\end{tabular}
\end{table}
\vspace{.3cm}
\section{Conclusion}

In this chapter, we have conducted a systematic investigation of multichannel convolution operators with a special
emphasis on the class of LSI Parseval operators.
What sets the Parseval operators apart from standard filterbanks is their lossless nature (energy conservation). This makes them ultra-stable and particularly easy to invert by mere flow-graph transposition of the computational architecture.
The other important feature is that the Parseval property is preserved through composition. Formally, this means that the Parseval filterbanks form a (non-commutative) operator algebra. On the more practical side, this enables the construction of higher-complexity filters through the chaining of elementary parametric modules, as exemplified in Section 4.

These properties make Parseval filterbanks especially attractive for the
design of robust (e.g, $1$-Lip) convolutional networks. We have demonstrated the application of such Parseval CNNs for the reconstruction of biomedical images. 
We have shown that the use of pre-trained Parseval filterbanks generally improves the
quality of iterative image reconstruction, while it offers the same mathematical guarantees as the conventional ``handcrafted'' reconstruction schemes. The training of such structures is straightforward---it is done before hand on a basic denoising task.
Further topics of research include (i) the investigation and comparison of different factorization schemes with the view of identifying the most effective ones, and (ii) the determination of the performance limits of CNN-based approaches under the mathematical constraint of stability/trustworthiness.

\bibliographystyle{abbrv}

\bibliography{/Users/unser/MyDrive/Bibliography/Bibtex_files/Unser}
%
%
\end{document}